\documentclass[11pt,letterpaper]{article}

\usepackage{amsmath}
\usepackage{amsthm}
\usepackage{amssymb}

\usepackage{fullpage}

\usepackage{setspace}
\usepackage[breaklinks]{hyperref}
\usepackage[usenames,dvipsnames]{xcolor}
\hypersetup{linktocpage=true,colorlinks=true,
            citebordercolor={.6 .6 .6},linkbordercolor={.6 .6 .6},%
citecolor=blue,linkcolor=BrickRed,urlcolor=black,pagecolor=black}
\usepackage[english]{babel}
\usepackage{graphicx}  
\allowdisplaybreaks

\usepackage{hyperref}
\usepackage{cleveref}
\usepackage{graphicx}
\usepackage{xcolor}
\usepackage{subfigure}
\usepackage{wrapfig}
\usepackage{thm-restate}
\usepackage{comment}

\usepackage{paralist}
\usepackage[ruled,noend]{algorithm2e}

\usepackage[shortlabels]{enumitem}

\newtheorem{theorem}{Theorem}[section]
\newtheorem{lemma}[theorem]{Lemma}
\newtheorem{assumption}[theorem]{Assumption}
\newtheorem{claim}[theorem]{Claim}
\newtheorem{corollary}[theorem]{Corollary}
\newtheorem{definition}[theorem]{Definition}
\newtheorem{proposition}[theorem]{Proposition}

\newif\ifFULL
\FULLtrue

\newcommand{\eps}{\varepsilon}
\newcommand{\sse}{\subseteq}

\newcommand{\I}{{\cal I}}
\newcommand{\poly}{\operatorname{poly}}

\newcommand{\opt}{{\mathsf opt}}

\newcommand{\cB}{{\cal B}}
\newcommand{\calS}{{\cal S}}
\newcommand{\cost}{{\mathsf {cost}}}

\newcommand{\tm}{{\widetilde m}}

\newcommand{\tM}{{\widetilde M}}

\newcommand{\tw}{{\widetilde w}}

\newcommand{\ells}{{\ell^\star}}

\newcommand{\ellone}{{\ell_1}}
\newcommand{\elltwo}{{\ell_2}}
\newcommand{\cA}{{\cal A}}
\newcommand{\tilt}{{\widetilde t}}

\newcommand{\tF}{{\widetilde F}}
\newcommand{\tf}{{\widetilde f}}
\newcommand{\vst}{{v^\star}}
\newcommand{\barp}{{\bar p}}
\newcommand{\const}{{64}}
\newcommand{\U}{U}

\newcommand{\IGNORE}[1]{}
\newcommand{\speed}{\sigma}
\renewcommand{\emptyset}{\varnothing}
\newcommand{\ones}[1]{\textbf{1}_{(#1)}}
\newcommand{\tsty}{\textstyle}

\newcommand{\mysubsection}[1]{\subsection{#1}}

\newcommand{\initOneLiners}{%
    \setlength{\itemsep}{0pt}
    \setlength{\parsep }{0pt}
    \setlength{\topsep }{0pt}
}
\newenvironment{OneLiners}[1][\ensuremath{\bullet}]
    {\begin{list}
        {#1}
        {\initOneLiners}}
    {\end{list}}

\title{Bag-of-Tasks Scheduling on Related Machines}

\author{ Anupam Gupta\thanks{
        (anupamg@cmu.edu)
        Computer Science Department,
        Carnegie Mellon University.
        }
    \and Amit Kumar\thanks{
        (amitk@cse.iitd.ac.in)
        Computer Science and Engineering Department,
        Indian Institute of Technology, Delhi.
        }
	\and Sahil Singla\thanks{
        (singla@cs.princeton.edu)
        Department of Computer Science,
        Princeton University.   
             }
    }

\begin{document}

\maketitle

\begin{abstract}
  We consider online scheduling to minimize weighted completion time
  on related machines, where each \emph{job} consists of several
  \emph{tasks} that can be concurrently executed.  A job gets
  {completed} when all its component tasks finish. We obtain an
  $O(K^3 \log^2 K)$-competitive algorithm in the
  \emph{non-clairvoyant} setting, where $K$ denotes the number of
  distinct machine speeds.  The analysis is based on dual-fitting on a
 precedence-constrained LP relaxation that may be of independent interest.
\end{abstract}

\section{Introduction} \label{sec:intro}

Scheduling to minimize the weighted completion time is a fundamental
problem in scheduling. Many algorithms have been developed in both the
online and offline settings, and for the cases where machines are
identical, related, or unrelated. Most of the work, however, focuses on
the setting where each job is a monolithic entity, and has to be
processed in a sequential manner.

In this work, we consider the online setting with multiple related
machines, where each \emph{job} consists of several {\em tasks}. These
tasks are independent of each other, and can be executed concurrently
on different machines. (Tasks can be preempted and migrated.) A job is said to have \emph{completed} 
 when all its component tasks finish processing. We consider the
\emph{non-clairvoyant} setting where the algorithm does not know the
size of a task up-front, but only when the task finishes 
processing. Such instances arise in operating system schedulers, where
a job and its tasks correspond to a process and its threads that can be
executed in parallel. This setting is sometimes called a \emph{``bag
  of tasks''} (see e.g.~\cite{4536445,MoschakisK15,BenoitMPRV10}). 

The bag-of-tasks model can be modeled using precedence
constraints. Indeed, each job is modeled as a star graph, where the
tasks correspond to the leaves (and have zero weight), and the root is
an auxiliary task with zero processing requirement but having weight
$w_j$. Hence the root can be processed only after all leaf tasks have
completed processing. The goal is to minimize total \emph{weighted
  completion time}.  Garg et al.~\cite{GGKS-ICALP19} gave a
constant-competitive algorithm for this problem for \emph{identical
  machines}, in a more general setting where tasks form arbitrary
precedence DAGs.

We extend this result to the setting of \emph{related machines} where
 machine $i$ has speed $s_i$. By losing a constant factor, we  assume that all speeds are powers of some constant $C$. Let $K$ denote the number of distinct machine speeds.  In~\S\ref{sec:stmt-lower-bound}, we show that this problem is
strictly more challenging than in the identical machines setting:

\begin{theorem}[Lower Bound]
  Any online non-clairvoyant algorithm has  $\Omega(K)$ competitive ratio for bags-of-tasks on
   related machines.
\end{theorem}

The lower bound %
arises because we want to
process larger tasks on fast machines, but we have no idea about the
sizes of the tasks, so we end up clogging the fast machines with small
tasks: this issue did not arise when machines were
identical. Given the lower bound, we now look for a non-clairvoyant scheduling algorithm
with a competitive ratio that depends %
on $K$, the number of distinct speeds. 
This number may be small in
many settings,  %
e.g., when we use commodity hardware of a limited number of types
(say, CPUs and GPUs). Our main result is a positive answer to this question:

\begin{theorem}[Upper Bound]
  The online non-clairvoyant algorithm for bags-of-tasks on
   related machines has a competitive ratio of \S\ref{sec:scheduling-algorithm} is
  $O(\min\{K^3 \log^2 K, K + \log n\})$.
\end{theorem}

Our algorithm uses a greedy strategy. Instead of explicitly building a
schedule, it assigns \emph{(processing) rates} to tasks at each time
$t$. Such a rate assignment is called feasible if for every $k$, the
rate assigned to any subset of $k$ tasks is at most the total speed of
the $k$ fastest machines. Using an argument based on Hall's matching theorem, a
schedule exists if and only if such a rate assignment can be found.
To assign these rates, each alive task gets a ``priority'', which
is the ratio of the weight of the job containing it to the number of
alive tasks of this job. In other words, a task with low weight  or
with many tasks gets a low priority. We assign feasible rates to
 alive tasks in a ``fair manner'', i.e., we cannot
increase the rate of a high priority task by decreasing the rate of a
lower priority task.  To efficiently find such feasible rates, we use a
water-filling procedure.

The analysis proceeds using the popular dual-fitting approach, but we need new ideas:
\textbf{(i)}~we adapt the precedence-constrained LP relaxation for completion time in~\cite{ChudakS99} to our setting. A naive relaxation would define the completion time of 
  a task as the maximum of the (fractional) completion times of each
  of the tasks, where the fractional completion time of a task is the
  sum over times $t$ of the fraction of the task remaining at this
  time
 Instead,  we define $U_{j,t}$, for a job $j$ and
  time $t$ as the maximum over all tasks $v$ for $j$ of the fraction
  of $v$ which remains to be completed at time $t$, the completion
  time of $j$ as $\sum_t U_{jt}$. (See \S\ref{sec:lp} for details.) 
 \textbf{(ii)}~Although it is natural to divide the machines into
classes based on their speeds, we need a finer partitioning, which
drives our setting of dual variables. Indeed, the usual idea of
dividing up the job's weight equally among the tasks that are still
alive only leads to an $O(\log n)$-competitiveness (see \S\ref{sec:weaker}). To do better, we
first preprocess the instance so that distinct machine speeds differ
by a constant factor, but the total processing capacity of a slower
speed class is far more than that of all faster machines. Now, at each
time, we divide the machines into {\em blocks}. A constant fraction of
the blocks have the property that either the average speed of the
machines in the block is close to one of the speed classes, or the
total processing  capacity of a block is close to that of {\em all}
the machines of a speed class. It turns out that our dual-fitting
approach works for accounting  the weight of jobs which get
processed by such blocks; proving this constitutes the bulk of technical part of the analysis. Finally, we show that most jobs (in terms of weight) get processed by such blocks, and hence we are able to bound the overall weighted completion time. 
We present the proofs in stages, giving intuition for the
new components in each of the sections.

\mysubsection{Related Work} Minimizing weighted completion time
on parallel machines with precedence constraints has
$O(1)$-approx\-imation in the \emph{offline} setting: Li~\cite{Li17}
improves on~\cite{Hall,MQS} to give a $3.387+\eps$-approximation. For
\emph{related} machines the precedence constraints make the problem
harder: there is an $O(\log m/\log \log
m)$-approximation~\cite{Li17} improving on a prior $O(\log K)$
result~\cite{ChudakS99}, and an $\omega(1)$ hardness  under
certain complexity assumptions~\cite{BNF15}. Here $m$ denotes the number of machines. These results are for offline and hence clairvoyant
settings, and  do not apply to our setting of non-clairvoyant
scheduling.

In the setting of parallel machines, there has been recent work on minimizing weighted completion time in DAG scheduling, where each job consists of a set of tasks with precedence constraints between them given by a DAG~\cite{RS,ALLM}. \cite{GGKS-ICALP19} generalized this to the non-clairvoyant setting and gave an $O(1)$-competitive algorithm. Our algorithm for the related case is based on a similar water-filling rate assignment idea. Since the machines have different speeds, a set of rates assigned to tasks need to satisfy a more involved feasibility condition. Consequently, its analysis becomes much harder; this forms the main technical contribution of the paper. Indeed, even for the special case considered in this paper where every DAG is a star, we can show a lower bound of $\Omega(K)$ on the competitive ratio of any non-clairvoyant algorithm. 
In \ifFULL \S\ref{sec:AppendixOne} \fi we show that any non-clairvoyant algorithm for related machines DAG scheduling must have $\Omega \big( \frac{\log m}{\log \log m} \big)$-competitive ratio. 

Our problem also has similarities to open shop scheduling. In
  open shop scheduling, each jobs consists of several tasks, where
  each task $v$ (for job $j$) needs to be processed on a distinct
  machine for $p_{vj}$ amount of time. However, unlike our setting,
  two tasks for a job cannot be processed simultaneously on different
  machines. \cite{QueyranneS01} considered open shop scheduling in the
  offline setting for related machines and gave a
  $(2+\eps)$-approximation. \cite{CorreaSV09} considered a further
  generalization of our problem to unrelated machines, where the tasks
  corresponding to distinct jobs need not be disjoint. They gave a
  constant-factor approximation algorithm, again offline. 

\mysubsection{Paper Organization} In this extended abstract, we first give
the algorithm in \S\ref{sec:scheduling-algorithm}, and the linear
program in \S\ref{sec:lp}. A simpler proof of
$O(K + \log n)$-competitiveness is in \S\ref{sec:weaker}. 
We show $\poly(K)$-competitiveness for the 
case of a single job (which corresponds to makespan minimization) in
\S\ref{sec:one-star}, and then give  the complete proof for the general case in
\S\ref{sec:stronger}. %

\section{Problem Statement and the $\Omega(K)$ Hardness}
\label{sec:stmt-lower-bound}

Each \emph{job} $j$ has a \emph{weight} $w_j$ and consists of
\emph{tasks} $T(j) = \{ (j,1), (j,2), \ldots, (j,k_j)\}$ for some
$k_j$. Each task $v =(j,\ell)$ has an associated \emph{processing
  requirement/size} $p_v = p_{(j,\ell)}$. The job $j$ completes when
all its associated tasks finish processing. We use letters $j, j'$,
etc.\ to denote jobs, and $v, v'$, etc.\ to denote tasks $(j,
\ell)$.

There are $m$ machines with speeds
$s_1 \geq s_2 \geq \ldots \geq s_m$. The goal is to minimize the
weighted completion time  of the jobs. We allow task \emph{preemption} and
\emph{migration}, and different tasks of a job can be processed
concurrently on different machines. However, a task itself can be
processed on at most one machine at any time.
In this extended abstract we consider the
special case when all release dates are 0,
but our results also extend to the more general setting 
of arbitrary release dates \ifFULL(see \S\ref{sec:genrj} for details)\fi.
Let $S_k := s_1 + \ldots + s_k$ denote the total speed of the fastest
$k$ machines. Since we care about the number of \emph{distinct}
speeds, we assume there are $K$ \emph{speed classes}, with speeds
$\speed_1 > \speed_2 > \ldots > \speed_K$. There are $m_i$ machines
having speed  $\speed_i$, where $\sum_i m_i = m$. 

\begin{assumption}[Increasing Capacity Assumption] \label{assump}  For parameter $\gamma \geq 1$:
  \begin{OneLiners}
  \item[(1)] \emph{(Falling Speeds.)} For each $\ell$, we have
    $\speed_i/\speed_{i+1} \geq \const$. 
  \item[(2)] \emph{(Increasing Capacity.)} For each $\ell$, the total processing
    capacity of speed class $\ell$ is at least twice that of the previous
    (faster) speed classes. I.e., %
    $ m_\ell \speed_\ell \geq 2(m_1 \speed_1 + \ldots + m_{\ell-1} \speed_{\ell-1}) . $
  \item[(3)] \emph{(Speed-up.)} The algorithm uses  machines that are
    $\gamma$ times faster than the adversary's machines.
  \end{OneLiners}
\end{assumption}

\begin{restatable}{proposition}{incap}
\label{prop:incap}
 An arbitrary instance can be transformed into one satisfying
  \Cref{assump} by losing a factor $O(\gamma K)$ in the competitive
  ratio. 
\end{restatable}

\begin{proof}{\em (Sketch)}
   For the first part, we round down the speed of each machine to a
  power of $\const$. This changes the completion time by at most a
  factor of $\const$.  The second increasing capacity assumption is
  not without loss of generality--- we greedily
  find a subset of speed classes  by losing
   $O( K)$ factor in  competitive ratio (see details in~\Cref{sec:missing2}).
   Finally, the $\gamma$-speedup can only change the competitive ratio by $\gamma$ factor. 
\end{proof}

Next we show that any online algorithm has to be $\Omega(K)$-competitive even for a single job with the machines satisfying increasing capacity \Cref{assump}.

\begin{proposition}
  \label{prop:lb}
  Any online algorithm is  $\Omega(K)$-competitive  even for a single job
  under increasing capacity \Cref{assump}.
\end{proposition}

\begin{proof}{\em(Sketch)}
  Consider a single job $j$ with $m$ tasks, where $m$ is the number of
  machines. For every speed class $\ell$, there are $m_\ell$ tasks of
  size $\sigma_\ell$---call these tasks $T_\ell(j).$ Since there is
  only one job, the objective is to minimize the makespan. The offline
  (clairvoyant) objective is $1$, since all tasks can be assigned to
  machines with matching speeds. However, any online algorithm incurs a
  makespan of $\Omega(K)$. 
 Here is an informal  argument, which can be proved even for
 randomized algorithms against oblivious adversaries: since there is no way to distinguish between
 the tasks, the algorithm can at best
  run all the alive tasks at the same speed. The tasks in $T_K(j)$
  will be the first to finish by time
  $ \frac{ m_K \sigma_K}{\sum_\ell m_\ell \sigma_\ell} \geq
  \frac{1}{2},$ where the inequality follows from the increasing
  capacity assumption. At this time, the processing on tasks from
  $T_\ell(j)$ for $\ell < K$ has been very small, and so tasks in
  $T_{K-1}(j)$ will require about $1/2$ more units of time to finish,
  and so on. %
\end{proof}

\section{The Scheduling Algorithm}
\label{sec:scheduling-algorithm}

The scheduling algorithm assigns, at each time $t$, a \emph{rate}
$L^t_v$ to each unfinished task $v$. The following lemma (whose proof is deferred to the appendix) characterizes
rates that correspond to schedules: %
\begin{restatable}{lemma}{feas}
  \label{lem:speedfeas}
  A schedule $\cal S$ is feasible %
  if for every time $t$ and every value of $k$: %
  \begin{quote}
    $(\star)$ the total rate assigned to any
    subset of $k$ tasks is at most $\gamma \cdot S_k$.
  \end{quote}
\end{restatable}

For each time $t$, we now specify  the rates $L^t_v$ assigned to each
unfinished task $v$.
For %
job $j$, let $T^t(j)$ be the set of tasks in $T(j)$
which are alive at time $t$. %
Initially all tasks are \emph{unfrozen}. %
We raise a parameter $\tau$, starting at zero, at a uniform speed. The
values taken by $\tau$ will be referred to as \emph{moments}. For
each job $j$ and each task $v \in T^t(j)$ that is unfrozen, define a
\emph{tentative rate} at  $\tau$ to be
\begin{gather}
  L^t_v := \frac{ w_j }{ | T^t(j) | } \cdot \tau \ . \label{eq:rate}
\end{gather}
Hence the tentative rates of these unfrozen tasks increase linearly,
as long as condition $(\star)$ is satisfied. However,
if $(\star)$ becomes tight for some
subset $V$ of alive tasks, i.e.,
$\tsty  \sum_{v \in V} L^t_v = \gamma \cdot S_{|V|}$, %
pick a \emph{maximal} set of such tasks and \emph{freeze} them, fixing
their rates at their current tentative values. (Observe the factor of
$\gamma$ appears on the right side because we assume the machines in
the algorithm to have a speedup of $\gamma$.) Now continue the
algorithm this way, raising $\tau$ and the $L^t_v$ values of remaining
unfrozen tasks $v$ until another subset gets tight, etc., stopping
when all jobs are frozen. This defines the $L^t_v$ rates for each task
$v$ for  time $t$. By construction, these rates satisfy~($\star$).

\mysubsection{Properties of the Rate Assignment}
\label{sec:prop-rate-assignm}

The following claim shows that all alive tasks corresponding to a job
get frozen simultaneously.

\begin{lemma}[Uniform Rates]
  \label{cl:rate}
  For any time $t$ and any job $j$, all its alive tasks (i.e., those
  in $T^t(j)$) freeze at the same moment $\tau$, and hence get the
  same rate.
\end{lemma}

\begin{proof}
  For the sake of contradiction, consider the first  moment $\tau$
  where a maximal set $V$ of tasks contains $v$ but not $v'$, for some
  job $j$ with $v, v' \in T(j)$. Both $v,v'$ have been treated
  identically until now, so $L^t_v = L^t_{v'}$. Also, by the choice of
  $\tau$, 
  $ \sum_{u \in V: u \neq v} L^t_u + L^t_v = \gamma S_{|V|} .$ 
  Since we maintain feasibility at all moments,
  \[ \tsty \sum_{u \in V: u \neq v} L^t_u + L^t_v + L^t_{v'} \leq \gamma S_{|V|+1}
    \qquad \text{and}
    \qquad \sum_{u \in V: u \neq v} L^t_u \leq \gamma S_{|V|-1} \ .
  \] 
  This implies $L^t_v \geq \gamma s_{|V|}$ and $L^t_{v'} \leq
  \gamma s_{|V|+1}$. Since $L^t_v = L^t_{v'}$ and $s_{|V|} \geq s_{|V|+1}$,
  all of these must be equal. In that case, by the maximality of set $V$, the algorithm should have picked $V \cup \{v'\}$
  instead of $V$.
\end{proof}

For a task $v \in T^t(j)$, define $\tw^t(v) := w_j/|T^t(j)|$ to be
task $v$'s ``share'' of the weight of job $j$ at time
$t$. So if task $v$ freezes at moment $\tau$, then its rate is $L_v^t = \tw^t(v)
\cdot \tau$.
Let us  
 relate this share for  $v$ to certain
averages of the weight. (Proof in \Cref{sec:missing3}) %

\begin{restatable}{corollary}{rate}
  \label{cor:rate1}
  Fix a time $t$. Let $V$ be the set of tasks frozen by some moment
  $\tau$. For a task $v \in V$, 
    \begin{OneLiners}
  \item[(i)] if   $V' \sse V$ is any subset of tasks
  which freeze either at the same moment as $v$, or after it, then
  $ \frac{\tw^t(v)}{ s_{|V|}} \geq \frac{w(V')}{S_{|V|}}\ .$

  \item[(ii)]  if $V'' \sse V$ is any subset of tasks which freeze either  at the same moment as $v$, or before it, then 
  $\frac{\tw^t(v)}{L^t_v} \leq \frac{\tw^t(V'')}{\sum_{v' \in V''} L^t_{v'}} \ . $
  \end{OneLiners}
\end{restatable}

\mysubsection{Defining the Blocks}

The rates for tasks alive at any time $t$ are defined by a sequence of
freezing steps, where some group of tasks are frozen: we call these
groups {\em blocks}. By \Cref{cl:rate}, all tasks in $T^t(j)$ belong
to the same block. The weight $w(B)$ of block $B$ is the total weight
of jobs whose tasks belong to $B$.  Let $B^t_1, B^t_2, \ldots $
be the blocks at time $t$ in the order they were frozen, and
$\tau^t_1, \tau^t_2, \ldots $ be the moments at which they froze. 
Letting $b^t_r:= |B^t_1 \cup \ldots \cup B^t_r|$, we get  that any task
$v \in B^t_r$ satisfies 
$\tau^t_v \cdot w(B^t_r) = \gamma(S_{b^t_{r+1}} - S_{b^t_r}) .$

Each block $B_r^t$ has an associated set of machines, namely the
machines on which the tasks in this block are processed---i.e., the
machines indexed $b^t_{r-1} + 1, \ldots, b^t_r$. We use $m(B)$ to
denote the set of machines associated with a block $B$. Since
$|B|=|m(B)|$ and the jobs in $B$ are processed on $m(B)$ in a
pre-emptive manner at time $t$, the rate assigned to any job is at
least the slowest speed (and at most the fastest speed) of the
machines in $m(B)$.

\section{The Analysis and Intuition}
\label{sec:lp}
We prove the competitiveness by a dual-fitting analysis: we give 
a primal-dual pair of LPs, use the algorithm above to give a feasible
primal, and then exhibit a feasible dual  with value
within a small factor of the primal cost.

In the primal LP, we
have variables $x_{ivt}$ for each task $v$, machine $i$, and time $t$ denoting the extent of processing done on
 task $v$ at machine $i$ during the interval $[t,t+1]$. Here $\U_{j,t}$
 denotes fraction of job $j$ finished at or after time $t$, and $C_j$
 denotes the completion time of job $j$. 
\begin{alignat}{2}
\min & \tsty \sum_j w_j C_j + \sum_{j,t} w_j \U_{j,t} &&\notag \\
\U_{j,t} & \geq \tsty \sum_{t' \geq t} \sum_i \frac{x_{ivt'}}{p_v} &\quad \quad &\forall j, \forall v \in T(j),  \forall t  \label{eq:deltatdual} \\
C_j & \geq  \tsty
\sum_{t} \sum_i 
 \frac{x_{ivt}}{s_i}  & &\forall j, \forall v \in T(j) \label{eq:deltaDual} \\
\tsty \sum_i \sum_{ t}
\frac{x_{ivt}}{p_v} & \geq 1 & &\forall j, \forall v \in T(j) \label{eq:alphaDual} \\
\tsty \sum_{v} \frac{x_{ivt}}{s_i} & \leq 1 & &\forall i, \forall t \label{eq:betaDual}
\end{alignat}
The constraint~\eqref{eq:deltatdual} is based on precedence-constrained LP relaxations for completion time. Indeed, each job can be thought of as a star graph with a zero size task at the root preceded by all the actual tasks at the leaf. 
%
  In our LP, for each time $t$, we define $U_{j,t} \in
[0,1]$ to be the maximum over
all tasks $v \in T(j)$ of the fraction of $v$ that remains (the RHS
of~\eqref{eq:deltatdual}), %
and the completion time of $j$ is at least the total sum over times
$t$ of $\U_{j,t}$ values. Since we do not explicitly enforce that
a task cannot be processed simultaneously on many machines, the first term $\sum_j w_j C_j$
is added to avoid a large integrality gap.
We show
feasibility of this LP relaxation (up to factor 2) in \S\ref{sec:appendixlp}.

\begin{restatable}{claim}{lp} \label{claim:feasibility}
For any schedule $\cal S$, there is a feasible solution to the LP of objective value at most  $2\,\cost(\calS).$
\end{restatable}

The linear programming dual has variables $\alpha_{j,v},
\delta_{j,v}, \delta_{j,v,t}$ corresponding to constraints
\eqref{eq:alphaDual},\eqref{eq:deltaDual},\eqref{eq:deltatdual} for every job $j$ and task $v \in T(j)$, and $\beta_{i,t}$ corresponding to constraints \eqref{eq:betaDual} for every machine $i$ and time $t$:
\begin{alignat}{2}
\notag \max. & \tsty \sum_{j,v} \alpha_{j,v} - \sum_{i,t} \beta_{i,t} &&\\
\label{eq:d1}
\tsty \frac{\alpha_{j,v}}{p_v} & \tsty \leq
 \frac{\beta_{i,t}}{s_i} + \sum_{t' \leq t} \frac{\delta_{j,v,t}}{p_v} + \frac{\delta_{j,v}}{s_i} &\quad  &\forall j, \forall i, \forall t , \forall v \in T(j) \\
\tsty \sum_{v \in T(j)} \delta_{j,v} & \leq w_j &&\forall j \label{eq:d2}  \\
\tsty \sum_{v \in T(j)} \delta_{j,v,t} & \leq w_j && \forall j, t \label{eq:d2'}
\end{alignat}

We now give some intuition about these dual variables. The
  quantity $\delta_{j,v,t}$ should be thought of the contribution (at
  time $t$) towards the weighted flow-time of $j$. Similarly,
  $\delta_{j,v}$ is {\em global} contribution of $v$ towards the
  flow-time of $v$. (In the integral case, $\delta_{j,v}$ would be
  $w_j$ for the task which finishes last. If there are several such
  tasks, $\delta_{j,v}$ would be non-zero only for such tasks only and
  would add up to $w_j$). The quantity $\alpha_{j,v}$ can be thought
  of as $v$'s contribution towards the total weighted flow-time, and $\beta_{i,t}$ is roughly the queue size at time $t$ on machine $i$. Constraint~\eqref{eq:d1} upper bounds $\alpha_{j,v}$ in terms of the other dual variables. More intuition about these variables can be found in \S\ref{sec:interpret}.

\mysubsection{Simplifying the dual LP} 

Before interpreting the dual variables, we rewrite the dual LP and
add some additional constraints. Define additional variables
$\alpha_{j,v,t}$ for each job $j$ and task $v \in T(j)$ and time $t$,
such that variable $\alpha_{j,v} = \sum_t \alpha_{j,v,t}$. We add a
new constraint:
\begin{align}
    \label{eq:dualsum}
\tsty     \sum_{v \in T(j)} \alpha_{j,v,t} \leq w_j. 
\end{align}
This condition is not a requirement in the dual LP, but we will set
 $\alpha_{j,v,t}$ to satisfy it. Assuming this, we set
$\delta_{j,v,t} := \alpha_{j,v,t}$ for all jobs $j$, tasks
$v \in T(j)$ and times $t$; feasibility of~\eqref{eq:dualsum} implies
that of~\eqref{eq:d2'}. Moreover,  \eqref{eq:d1} simplifies
to
\[ \tsty  \sum_{t' \geq t} \frac{\alpha_{j,v,t'}}{p_v} \leq
  \frac{\beta_{i,t}}{s_i} + \frac{\delta_{j,v}}{s_i}. \] %
  Observe that we can
write $p_v$ as the sum of the rates, and hence as
$p_v = \sum_{t'} L^{t'}_v$. Since this is at least
$\sum_{t' \geq t} L^{t'}_v$ for any $t$, we can substitute above, and
infer that it suffices to verify the
following condition for all tasks $v \in T(j)$, time $t$, and time
$t' \geq t$:
\begin{align}
    \label{eq:dualnew}
\tsty     \alpha_{j,v,t'} \leq \frac{\beta_{i,t} \cdot L^{t'}_v}{s_i}  + \frac{\delta_{j,v} \cdot L^{t'}_v}{s_i}.
\end{align}
Henceforth, we  ensure that our duals (including $\alpha_{j,v,t}$) satisfy~\eqref{eq:dualsum},\eqref{eq:dualnew} and~\eqref{eq:d2}.

\mysubsection{Interpreting the Duals and the High-Level Proof Idea}
\label{sec:interpret}

We give some intuition about the dual variables, which will be useful
for understanding the subsequent analysis. We set dual variables
$\alpha_{j,v}$ such that for any job $j$, the sum   $\sum_{v \in T(j)}
\alpha_{j,v}$ is (approximately) the weighted completion of job
$j$. This ensures that $\sum_{j,v} \alpha_{j,v}$ is the total weighted
completion of the jobs. One way of achieving this is as follows: for
every time $t$ and task-job pair $(j,v)$ we define $\alpha_{j,v,t}$ variables such  that they add up to
be $w_j$ if job $j$ is unfinished at time $t$ (i.e., \eqref{eq:dualsum} is satisfied with equality). If $\alpha_{j,v}$ is
set to $\sum_t \alpha_{j,v,t}$, then these $\alpha_{j,v}$ variables
would add up to the weighted completion time of $j$.

The natural way of defining $\alpha_{j,v,t}$ is to evenly distribute
the weight of $j$ among all the alive tasks at time $t$, i.e., to set
$\alpha_{j,v,t} = \frac{w_j}{T^t(j)}$.  This idea works if we only
want to show that the algorithm is $O(\log n)$-competitive, but does
not seem to generalize if we want to show $O(K)$-competitiveness. The
reason for this will be clearer shortly, when we discuss the
$\delta_{j,v}$ variables.

Now we discuss $\beta_{i,t}$ dual variables. We set these variables so
that $\sum_{t} \beta_{i,t}$ is a constant (less than $1$) times the
total weighted completion time. This ensures that the objective value of the
dual LP is also a constant times the total weighted completion time. A
natural idea (ignoring constant factors for now) is to set
$\beta_{i,t} = \frac{w(A^t)}{K m_\ell}$, where $A^t$ is the set of
alive jobs at time $t$ and $\ell$ is the speed class of machine
$i$. Since we have put an $\Omega(K)$ term in the denominator of
$\beta_{i,t}$ (and no such term in the definition of $\alpha_{j,v}$),
ensuring the feasibility of~\eqref{eq:d1} would require a speed
augmentation of $\Omega(K)$.

Finally, consider the $\delta_{j,v}$ dual variables. As~\eqref{eq:d2}
suggests, setting $\delta_{j,v}$ is the same as deciding how to
distribute the weight $w_j$ among the tasks in $T(j)$. Notice,
however, that this distribution cannot depend on time (unlike
$\alpha_{j,v,t}$ where we were distributing $w_j$ among all the alive
tasks at time $t$). In the ideal scenario, tasks finishing later
should get high $\delta_{j,v}$ values. Since we are in the
non-clairvoyant setting, we may want to set
$\delta_{j,v} = \frac{w_j}{|T(j)|}$. We now argue this can lead to a
problem in satisfying~\eqref{eq:dualnew}.

Consider the setting of a single unit-weight job $j$ initially
having $n$ tasks, and so we set $\delta_{j,v} = \frac{1}{n}$ for all
$v$. Say that $n = m_\ell$ for a large value of $\ell$: by the
increasing capacity assumption, $m_\ell \approx m_1 + \ldots + m_\ell$.
Now consider a later point in time $t$ when only $n'$ tasks remain,
where $n' = m_{\ell'}$ for some speed class $\ell' \ll \ell$. At this
time $t$, each of the $n'$ surviving tasks have
$\alpha_{j,v,t} = \frac{1}{n'}.$ But look at the RHS of~\eqref{eq:dualnew},
with machine $i$ of speed class $\ell$. The rate $L^t_v$ will be very
close to $\speed_{\ell'}$ (again, by the increasing capacity
assumption), and so both the terms would be about
$\frac{\speed_{\ell'}}{m_\ell \speed_\ell}$. However,
$m_\ell \speed_\ell$ could be much larger than
$m_{\ell'} \speed_{\ell'}$, and so this constraint will not be
satisfied. In fact, we can hope to satisfy~\eqref{eq:dualnew} at some time
$t$ only if $n'$ is close to $n$, say at least $n/2$. When the number
of alive tasks drops below $n/2$, we need to {\em redistribute} the
weight of $j$ among these tasks, i.e., we need to increase the
$\delta_{j,v}$ value for these tasks, to about $\frac{1}{n/2}$. Since these
halving can happen for $\log n$ steps, we see that~\eqref{eq:deltaDual} is violated by a factor of $\log
n$. These ideas can be extended to give an $O(\log n + K)$-competitive
algorithm for arbitrary inputs; see  \S\ref{sec:weaker}
for details. To get a better bound, we need a more careful setting of
the dual variables, which we talk about in \S\ref{sec:one-star} and \S\ref{sec:stronger}.

\section{Analysis I: A Weaker $O(K + \log n)$ Guarantee}
\label{sec:weaker}

We start with a simpler analysis which yields an
$O(K + \log n)$-competitiveness. This argument will not use the
increasing capacity assumption from \Cref{assump}; however, the result
gives a competitiveness of $O(\max(K, \log n))$ which is logarithmic
when $K$ is small, whereas our eventual result will be
$O(\min(K^{O(1)}, K + \log n))$, which can be much smaller when $K \ll
\log n$.

\begin{theorem} \label{thm:weakerGuar}
  The scheduling algorithm in \S\ref{sec:scheduling-algorithm} is $O(K + \log n )$-competitive.
\end{theorem}

\begin{proof}
For each job $j$, we arrange the
tasks in $T(j)$ in descending order of their processing requirements.
(This is the opposite of the order in which they finish, since all alive
tasks of a job are processed at the same rate.)  Say the sequence of
the tasks for a job $j$ is $v_1, \ldots, v_{r}$. We partition these tasks
into \emph{groups} with exponentially increasing cardinalities:
$T_1(j) := \{v_1\}$, $T_2(j) := \{v_2, v_3\}$, and
$T_h(j) := \{v_{2^{h-1}}, \ldots, v_{2^h-1}\}$ has $2^{h-1}$
tasks.  (Assume w.l.o.g.\ that $r+1$ is a power of $2$ by adding
 zero-sized tasks to $T(j)$).
Now we define the dual variables.

\medskip

\noindent\textbf{Dual Variables.}  Define $\gamma := 2\max\{K, \log_2 n\}$.
\begin{itemize}
\item For a time $t$ and machine $i$ of speed class $\ell$, let $A^t$
  denote the set of active (unfinished) jobs at time $t$, and define
  $\displaystyle \beta_{i,t} := \frac{w(A^t)}{m_\ell \cdot \gamma} \ .$
\item For job $j$ and a task $v \in T_h(j)$ in the $h$-th group,
  define $ \displaystyle \delta_{j,v} := \frac{w_j}{2^{h-1} \cdot \gamma} \ .$
  
\item In order to define $\alpha_{j,v}$, we first define quantities
  $\alpha_{j,v,t}$ for every time $t$, and then set
  $\alpha_{j,v} := \sum_t \alpha_{j,v,t}$. At time $t$, recall that
  $T^t(j)$ is the set of alive tasks of job $j$, and define
 $$ \displaystyle \alpha_{j,v,t} ~:=~ \frac{w_j}{|T^t(j)|} \cdot \ones{ v \text{ alive
        at time $t$}} ~=~ \frac{w_j}{|T^t(j)|} \cdot \ones{ v \in
      T^t(j)} \ . $$
  This ``spreads'' the weight of $j$ equally among its alive tasks.
\end{itemize}

Having defined the dual variables, we first argue that they are feasible.

 \begin{lemma} [Dual feasibility]
 \label{lem:d1special}
 The dual variables defined above always satisfy the constraints~\eqref{eq:dualsum},~\eqref{eq:d2} and\eqref{eq:dualnew}  
for a speed-up factor $\gamma \geq 2\max\{K, \log_2 n\}$.
 \end{lemma}
 \begin{proof}
 To check feasibility of~(\ref{eq:d2}), consider a job
 $j$ and observe that
 $$ \tsty \sum_{v \in T(j)} \delta_{j,v} ~~=~~ \sum_h \sum_{v \in T_h(j)} \delta_{j,v} 
 ~~=~~\sum_h \sum_{v \in T_h(j)} \frac{w_j}{2^{h-1} \cdot \gamma} ~~=~~ \sum_h \frac{w_j}{\gamma} ~~\leq~~ w_j \ ,$$
 because $|T_h(j)| = 2^{h-1}$  and there are at most $\log_2 n \leq \gamma$ distinct groups. Feasibility of~\eqref{eq:dualsum} also follows easily. It remains to check~\eqref{eq:dualnew} for a job $j$, task $v$, machine $i$ and times $t' \leq t$. 

  If $v$ is not alive at time $t'$, then $\alpha_{j,v,t'}$ is 0,
  and~(\ref{eq:dualnew}) follows trivially. Else, $v \in T^{t'}(j)$, and
  suppose $v \in T_h(j)$. This means the jobs in
  $T_1(j), \ldots, T_{h-1}(j)$ are also alive at time $t'$, so
  $|T^{t'}(j)| \geq 1 + 2 + \ldots + 2^{h-2} + 1 = 2^{h-1}.$
  Furthermore, suppose the tasks in $T^{t'}(j)$ belong to block $B$
  (defined in \S\ref{sec:prop-rate-assignm}), and let $\ell^\star$ be
  the speed class with the slowest machines among the associated
  machines $m(B)$. Let $\ell$ denote the speed class of machine~$i$
  (considered in~\eqref{eq:dualnew}). Two cases arise:
  the first is when
        $\ell \geq \ells$, %
       where $L_v^{t'} \geq
       \gamma \speed_{\ell^\star} \geq \gamma \speed_\ell = \gamma s_i$, so \eqref{eq:dualnew} holds because %
      $$ \tsty \alpha_{j,v,t'} ~~=~~ \frac{w_j}{|T^{t'}(j)|} ~~\leq~~
      \frac{w_j}{2^{h-1}} ~~=~~ \gamma\cdot \delta_{j,v}
      ~~\leq~~ \frac{\delta_{j,v} \cdot L^{t'}_v}{s_i} \ .$$

  The second case is $\ell < \ells$: Let $V \sse A^{t'}$ be the set of jobs which are frozen
      by the moment $v$ freezes. In other words, $V$ contains 
      tasks in block $B$ and the blocks before it. Applying the second
      statement in \Cref{cor:rate1} with $V'' = V$,
      $$ \tsty \frac{w_j}{|T^{t'}(v)| L^{t'}_v} ~~\leq~~  \frac{\tw^t (V)}{\sum_{v' \in V} L^{t'}_{v'}} ~~\leq~~ \frac{w(V)}{\sum_{v' \in V} L^{t'}_{v'}} ~~\leq~~ \frac{w(A^{t'})}{\gamma \cdot m_\ell \speed_\ell} \ , $$
      where the last inequality uses the fact that all machines of speed class $\ell$ are busy processing jobs in $V$. Therefore, 
      $$ \tsty \alpha_{j,v,t'} ~~=~~ \frac{w_j}{|T^{t'}(j)|}
       ~~\leq~~ \frac{w(A^{t'}) \cdot L^{t'}_v}{\gamma \cdot m_\ell \sigma_\ell} ~~\leq~~ 
       \frac{ \beta_{i,t} \cdot L^{t'}_v}{s_i}\  ,$$
       the last inequality useing the definition of $\beta_{i,t}$ and
        that $w(A^t) \geq w(A^{t'}).$ 
 \end{proof}
 
 Finally, we show that the dual objective value for this setting of
 dual variables is close to the primal value. It is easy to check that
 $\sum_j \alpha_j = \sum_t w(A^t)$, which is the total weighted
 completion time of the jobs. Moreover,
 $$ \tsty \sum_{i,t} \beta_{i,t} ~~=~~ \sum_t \sum_{\ell} \sum_{i: s_i =
   \speed_\ell} \frac{w(A^t)}{m_\ell \; \gamma} ~~=~~ \frac{K}{\gamma}
 \cdot \sum_t w(A^t) \ .$$ 
 Since we chose  speedup  $\gamma = 2\max\{K, \log_2 n\}$, we have
 $K \leq \gamma/2$ and the dual objective value
 $\sum_{j,v} \alpha_{j,v} - \sum_{i,t} \beta_{i,t}$ is at least half
 of the total weighted completion time (primal
 value). This completes the proof of \Cref{thm:weakerGuar}. 
 \end{proof}

\section{Analysis II: An Improved Guarantee for a Single Job}
\label{sec:one-star}
 
We want to show that the competitiveness of our algorithm just depends
on $K$, the number of speed classes. To warm up, in this section we consider the
special case of a single job; in \S\ref{sec:stronger}
we consider the general case.
As was shown in \Cref{prop:lb}, 
any algorithm has competitive ratio $\Omega(K)$ even in
the case of a single job. We  give a
matching upper bound using dual fitting 
for an
instance \textit{with a single job} $j$, say of weight $1$, when the machines satisfy \Cref{assump}. 

\begin{theorem} \label{thm:singleJob}
  If the machines satisfy~\Cref{assump}, the scheduling algorithm in
  \S\ref{sec:scheduling-algorithm} is $O(K^2)$-competitive for a single job.
\end{theorem}
\mysubsection{The Intuition Behind the Improvement}

The analysis in \S\ref{sec:weaker} incurred $\Omega(\log
n)$-competitive ratio because we divided the execution of the tasks of
each job into $O(\log n)$ {\em epochs}, where each epoch ended when
the number of tasks
halved. %
In each such epoch, we set the $\delta_{j,v}$ variables by
distributing the job's weight evenly among all tasks alive at the
beginning of the epoch.
A different way to define epochs would be to let them correspond to
the time periods when the number of alive tasks falls in the range
$(m_\ell, m_{\ell+1})$. This would give us only $K$ epochs. There is a
problem with this definition: 
as the number of tasks vary in the range $(m_\ell, m_{\ell+1})$, the
rate assigned to tasks varies from $\speed_\ell$ to
$\speed_{\ell+1}$. Indeed, there is a transition point $\tm_\ell$ in
$(m_\ell, m_{\ell+1})$ such that the rate assigned to the tasks stays
close to $\speed_{\ell+1}$ as long as the number of tasks lie in the
range $(\tm_\ell, \speed_{\ell+1})$; but if the number of tasks lie in
the range $(m_\ell, \tm_{\ell})$, the assigned rate may not stay close
to any fixed value. However, in this range, the {\em total} processing
rate assigned to all the tasks stays close to $m_\ell \speed_\ell$.

It turns out that our argument for an epoch (with minor modifications)
works as long as one of these two facts hold during an epoch: (i) the
total rate assigned to the tasks stays close to $m_\ell \speed_\ell$
for some speed class $\ell$ (even though the number of tasks is much
larger than $m_\ell$), or (ii) the actual rate assigned to the tasks
stays close to $\speed_\ell$. Thus we can divide the execution of the
job into $2K$ epochs, and get an $O(K)$-competitive algorithm. In this
section, we prove this for a single job; we extend to the case of
multiple jobs in \S\ref{sec:stronger} (with a slightly worse competitiveness).

\mysubsection{Defining the New Epochs}

Before defining the dual variables, we begin with a
definition.  For each speed class $\ell$, define the \emph{threshold}
$\tm_\ell$ to be the following: 
\begin{gather}
\textstyle   \tm_\ell := \frac{1}{\speed_{\ell+1}} \left( \speed_1 m_1 + \dots + \speed_\ell
  m_\ell \right). \label{eq:threshold}
\end{gather}
The parameter $\tm_\ell$ is such that the processing capacity of
$\tm_\ell$ machines of class $\ell+1$ equals the combined processing
capacity of machines of class at most $\ell$.
The increasing capacity assumption  implies  $m_\ell <  \tm_\ell < m_{\ell+1}$, as formalized   below:

\begin{claim}
  \label{cl:mlt}
  Define $M_\ell := m_1 + \ldots + m_\ell$ and
  $\tM_\ell := M_\ell + \tm_\ell$.  Under the increasing capacity \Cref{assump}
  and $\kappa=2$, for any speed class $\ell$, we have
  \begin{OneLiners}
  \item[(a)] $2 \tm_\ell \leq m_{\ell+1}$ and so, $\tM_{\ell} \leq M_{\ell+1}$, \qquad 
  (b) $\tM_\ell \geq 2M_\ell$, 
  \item[(c)] $ m_\ell \sigma_\ell \geq \frac{1}{2} \tm_{\ell} \sigma_{\ell+1}$, and \qquad
 (d) $\tm_{\ell} \geq 2m_\ell$.
  \end{OneLiners}
\end{claim}
\begin{proof}
  Fact~(a) follows from the increasing capacity assumption  and the
  definition of the threshold, since
 $2 \tm_\ell \sigma_{\ell+1} \leq \sigma_{\ell+1} m_{\ell+1}$. This implies
 $\tM_{\ell} = M_\ell+\tm_\ell \leq M_\ell+m_{\ell+1}  \leq M_{\ell+1}$.
 Proving~(b) is equivalent to showing
  $\tm_\ell \geq M_\ell$, which follows from the definition of 
   $\tm_\ell$ and the fact that $\speed_{\ell+1} < \speed_i$
  for all $i \leq \ell$. The last two statements also follow 
  from the increasing capacity assumption.
\end{proof}

\begin{figure}
    \centering
    \includegraphics[width=4.2in]{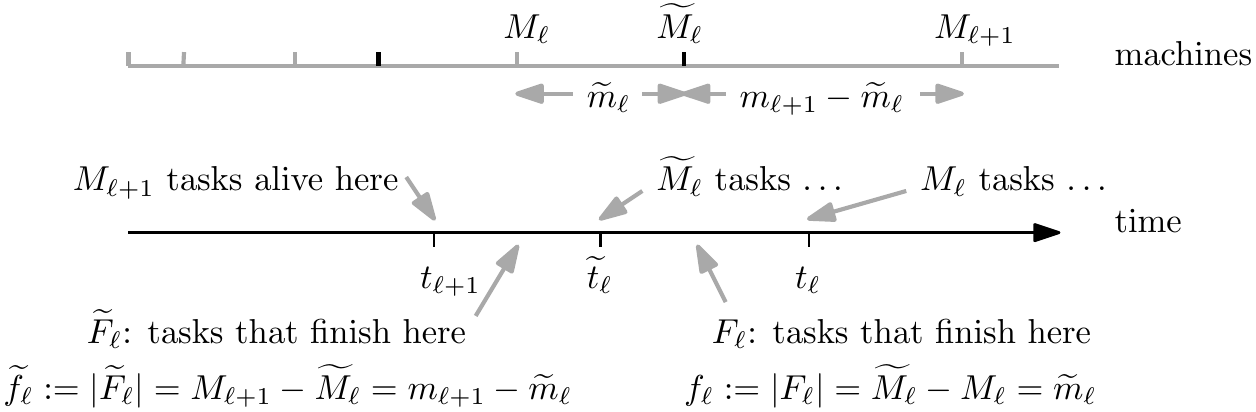}
    \caption{\small\emph{Defining breakpoints. }}
    \label{fig:timeline}
\end{figure}

We identify a set of $2K$ \emph{break-points} as follows: for each speed
class $\ell$, let $t_\ell$ denote the first time when $M_\ell$ alive
tasks remain. Similarly, let $\tilt_\ell$ be the first time when
exactly $\tM_\ell$ alive tasks remain. Note that
$t_{\ell+1} < \tilt_\ell < t_{\ell}$. Let $\tF_\ell$ be the tasks which
finish during $[t_{\ell+1},\tilt_\ell]$, and $F_\ell$ be those which
finish during $[\tilt_\ell,t_\ell]$. Let  $\tf_\ell$ and $f_\ell$ 
denote the cardinality of $\tF_\ell$ and $F_\ell$ respectively.  Note
that
$\tf_\ell= M_{\ell+1} -
\tM_\ell = m_{\ell+1} - \tm_{\ell}, f_\ell = \tM_\ell-M_\ell = \tm_\ell$.

\begin{claim}
\label{cl:fl}
For any speed class $\ell$, we have $f_\ell \leq \tf_\ell \leq f_{\ell+1}.$
\end{claim}
\begin{proof}
  The first statement requires that
  $\tm_\ell \leq m_{\ell+1} - \tm_{\ell}$. This is the same as
  $2 \tm_\ell \leq m_{\ell+1}$, which follows from \Cref{cl:mlt}\,(a).
  The second statement requires that
  $m_{\ell+1} - \tm_{\ell} \leq \tm_{\ell+1},$ i.e.,
  $m_{\ell+1} \leq \tm_{\ell} + \tm_{\ell+1}. $ But
  $m_{\ell+1} \leq \tm_{\ell+1}$ (by \Cref{cl:mlt}\,(d)), hence the proof.
\end{proof}

Next we set the duals. Although it is possible to directly argue that the delay incurred by the job in each epoch is at most (a constant times) the optimal objective value,  the dual fitting proof generalizes to the arbitrary set of jobs. 
\mysubsection{Setting the Duals}

Define the speed-up $\gamma \geq
2K$. We set the duals as:
\begin{itemize}
\item Define 
$\delta_{j,v} := \left\{
    \begin{array}{ll} \frac{1}{2K \cdot f_\ell}
      & \mbox{if $v \in F_\ell$} \\ \frac{1}{2K
      \cdot \tf_\ell} & \mbox{if $v \in
                        \tF_\ell$} \end{array}
                  \right. \ .
  $
\item For machine $i$  of class
$\ell$, define
 $  \beta_{i,t} := \frac{1}{2K \cdot m_\ell}\cdot \mathbf{1}_{(\text{not
      all tasks finished})} \ . $
\item Finally, as in \S\ref{sec:weaker}, we define $\alpha_{j,v,t}$ for each
task $v$ of job $j$, and then set
$\alpha_{j,v} := \sum_t \alpha_{j,v,t}$.  
 To define $\alpha_{j,v,t}$, we consider two cases (we use $n_t$
to denote the number of alive tasks at time $t$):
\begin{enumerate}
\item 
 $n_t \in [M_\ell, \tM_\ell)$ for some $\ell$:
 Then~~
 $\alpha_{j,v,t} := (1/n_t)\cdot \ones{ v \text{ alive
        at time $t$}}.$
        
\item %
$n_t \in [\tM_\ell, M_{\ell+1})$ for some $\ell$:
Then~~
$ \alpha_{j,v,t} := (1/f_\ell)\cdot \ones{ v \in F_{\ell} }. $
\end{enumerate}
Note the asymmetry in the definition. It  arises
because in the first case, the {\em total} speed of  machines
processing a task is (up to a constant) $m_\ell \sigma_\ell$, whereas
in the second case the {\em average} speed of such machines is about
$\sigma_{\ell+1}$.

\end{itemize}

 \begin{lemma} [Dual feasibility] \label{lem:dualFeasSingle}
 The dual variables defined above always satisfy the
 constraints~(\ref{eq:d2}) and~\eqref{eq:dualsum},  and satisfy 
 constraint~(\ref{eq:dualnew}) for speed-up $\gamma \geq 2K$.
 \end{lemma}
\begin{proof}
It is easy to check from the definition of $\delta_{j,v}$ and $\alpha_{j,v,t}$ that the dual constraints~\eqref{eq:d2} and~\eqref{eq:dualsum} are satisfied. 
It remains to verify constraint~\eqref{eq:dualnew} (re-written below) for any task $v$, machine $i$, times
$t$ and $t' \geq t$. 
\begin{gather*}
  \alpha_{j,v,t'} \leq \frac{L^{t'}_v}{s_i} \cdot (\beta_{i,t} +
  \delta_{j,v}) \ .
     \tag{\eqref{eq:dualnew} repeated}
\end{gather*}
As in the definition of $\alpha_{j,v,t'}$, there are two cases
depending on where $n_{t'}$ lies. First assume that there is class
$\ells$ such that $M_\ells \leq n_{t'} < \tM_\ells$. Assume that $v$
is alive at time $t'$ (otherwise $\alpha_{j,v,t'}$ is 0), so
$\alpha_{j,v,t'} = \frac{1}{n_{t'}}$, where $n_{t'}$ is the number of
alive tasks at time $t'$. Being alive at this time $t'$, we know that $v$
will eventually belong to some $F_{\ell}$ with $\ell \leq \ells$, or in
some $\tF_{\ell}$ with $\ell < \ells$. So by \Cref{cl:fl},
$\delta_{j,v} \geq \frac{1}{2K \cdot f_{\ells}}$. Moreover, let $i$ be
a machine of some class $\ell$, so $s_i = \sigma_\ell$. Hence, it is
enough to verify the following in order to satisfy~\eqref{eq:dualnew}:
\begin{eqnarray}
\label{eq:pf2}
\frac{1}{n_{t'}} ~~\leq~~ \frac{L^{t'}_v }{\speed_\ell} \cdot
  \bigg(\frac{1}{2K\cdot m_\ell} +  \frac{1}{2K \cdot f_{\ells}}\bigg) \  .
\end{eqnarray}
Two subcases arise, depending on how $\ell$ and $\ells$ relate---in
each we show that just one of the terms on the right is larger than
the left.
\begin{itemize}
    \item $\ells \geq \ell$: Since at least $M_{\ells} $ tasks are
      alive at this time, the total speed assigned to all the alive
      tasks at time $t'$ is at least $\gamma \cdot \sigma_\ells
      m_\ells$. Therefore, $L^{t'}_v \geq \frac{\gamma \cdot m_\ells
        \sigma_\ells}{n_{t'}}$. Now using $\gamma \geq 2K$, we get
    $$ \frac{ L^{t'}_v }{2K \cdot m_\ell \sigma_\ell} ~~\geq~~ 
    \frac{m_\ells \sigma_\ells}{m_\ell \sigma_\ell}\cdot \frac{1}{n_{t'}} ~~\geq~~ \frac{1}{n_{t'}} \ , $$
    where the last inequality follows from the increasing capacity
    assumption.
    \item $\ells \leq \ell-1$: The quantity $L^{t'}_v n_{t'}$ is the
      total speed of the machines which are busy at time $t'$, which
      is at least $\gamma(m_1 \sigma_1 + \ldots + m_\ells \sigma_\ells) = \gamma \cdot 
      \tm_\ells \sigma_{\ells+1}.$ Again, using $\gamma \geq 2K$, we get
    $$\frac{ L^{t'}_v \cdot n_{t'} }{2K \cdot f_\ells
      \sigma_\ell}  ~~\geq~~  \frac{\tm_\ells \sigma_{\ells+1}}{f_\ells \sigma_\ell} ~~\geq~~ 1 $$
    because $\sigma_{\ells+1} \geq \sigma_\ell$ and $\tm_\ells = f_\ells.$
\end{itemize}
Thus, \eqref{eq:pf2} is satisfied in both the above subcases.

Next we consider the case when there is a speed class $\ells$ such that
$\tM_\ells < n_{t'} \leq M_{\ells+1}. $ We can assume that
$v \in F_\ells$, otherwise $\alpha_{j,v,t'}$ is 0; this means
$\delta_{v,j} = \frac{1}{2K \cdot f_\ells}$. Since
$\alpha_{j,v,t'} = \frac{1}{f_\ell} = \frac{1}{\tm_\ells}$, and
$L^{t'}_v \geq \gamma \cdot  \sigma_{\ells+1}$, the expression~(\ref{eq:dualnew})
follows from showing
\begin{eqnarray}
\label{eq:pf3}
\frac{1}{\tm_\ells} ~~\leq~~ \frac{\gamma}{\sigma_{\ell}} \cdot \bigg(\frac1{2K \cdot
  m_\ell} + \frac{1}{2K \cdot f_\ells} \bigg) \cdot \sigma_{\ells+1} \ .
\end{eqnarray}
Since $\gamma \geq 2K$, we can drop those terms. Again, two cases arise: 
\begin{itemize}
    \item $\ells \geq \ell$: By definition, $\sigma_{\ells+1} \cdot \tm_{\ells} 
    \geq \sigma_\ells m_\ells \geq \sigma_\ell m_\ell$ (by the increasing capacity assumption).
    \item $\ells \leq \ell-1$: Since $f_\ells = \tm_\ells$ and
      $\sigma_\ell \leq \sigma_{\ells+1},$ this case also follows
      easily. \qedhere  
\end{itemize}
\end{proof}

\begin{proof}[Proof of \Cref{thm:singleJob}]
Having checked dual feasibility in \Cref{lem:dualFeasSingle}, 
consider now the objective function. For
any time $t$ when at least one task is alive,
$\sum_v \alpha_{j,v,t} = 1$. Therefore, $\sum_{v} \alpha_{j,v}$ is the
makespan. Also, $\sum_{i} \beta_{i,t} = 1/2$ as long as there are
unfinished tasks, so $\sum_{i,t} \beta_{i,t}$ is half the makespan,
and the objective function
$\sum_v \alpha_{j,v} - \sum_{i,t} \beta_{it}$ also equals half the
makespan. Since we had assumed $\gamma = O(K)$-speedup, the algorithm
is $O(K)$-competitive.
 \end{proof}

\section{Analysis III: Proof  for $\tilde{O}(K^3)$ Guarantee}  

\label{sec:stronger} 

\newcommand{\last}{\tau}

We now extend the ideas from the  single job case to the general
case. 
For time $t$, let
$A^t$ be the set of alive jobs at time $t$. Unlike the single job case where we had only one block, we can now have multiple blocks. While defining $\alpha_{j,v,t}$ in the single job case, we had considered two cases:  (i)~the rate assigned to each task stayed close to $\speed_\ell$ for some class $\ell$ (this corresponded to $n_t \in [\tM_{\ell-1}, M_{\ell})$), and (ii)~the total rate assigned to each task was close to $m_\ell \speed_\ell$ for speed class $\ell$ (this corresponded to $n_t \in [M_\ell, \tM_\ell)$). We extend these notions to blocks as follows: 

\medskip
     \noindent {\em Simple blocks:} A block $B$ is said to be {\em simple} w.r.t. to a speed class $\ell$ if the average rate assigned to the tasks in $B$ is close to $\speed_\ell$. Similarly a job $j$ is said to be simple w.r.t. a speed class $\ell$ if all the alive tasks in it are assigned rates close to $\speed_\ell$ (recall that all alive tasks in a job are processed at  the same rate). All the jobs in a simple block  $B$ may not be simple (w.r.t. the same speed class $\ell$), but we show that a large fraction of jobs  (in terms of weight) in $B$ will be simple\ifFULL (\Cref{lem:simpleb})\else\fi. Thus, it is enough to account for the weight of simple jobs in $B$. This is analogous to case~(i) mentioned above (when there is only one job and tasks in it receive rate close to $\speed_\ell$). In~\S\ref{sec:one-star}, we had defined $\alpha_{j,v,t}$ for such time $t$ as follows: we consider only those tasks which survive in $F_\ell$, and then evenly distribute $w_j$ among these tasks. The analogous definition here would be as follows: let $\tau_{\ell,j}$  be the {\em last} time when $j$ is simple w.r.t. the speed class $\ell$. We define $\alpha_{j,v,t}$ by evenly distributing $w_j$ among those tasks in $v$ which are alive at $\tau_{\ell,j}$. We give details \ifFULL in~\S\ref{sec:simpledual}. \fi
    
     \medskip
     \noindent{\em Long blocks:}  The total speed of the machines in this block stays close to $m_\ell \speed_\ell$ for some speed class $\ell$. Again, inspired by the definitions in~\S\ref{sec:one-star}, we  assign $\alpha_{j,v,t}$ for tasks $v \in B$ by distributing $w(B)$ to these tasks (in proportion to the rate assigned to them). From the perspective of a job $j$ which belongs to a long block $B$ w.r.t. a speed class $\speed_\ell$ at a time $t$, the feasibility of~\eqref{eq:d1} works out provided for {\em all} subsequent times $t'$ when $j$ again belongs to such a block $B'$, we have $w(B')$ and $w(B)$ remain close to each other. If $w(B')$ exceeds (say) $2w(B)$, we need to reassign a new set of $\delta_{j,v}$ values for $v$. To get around this problem  we require that long blocks (at a time $t$) also have weight at least $w(A^t)/(10K)$. With this requirement, the doubling step mentioned above can only happen $O(\log K)$ times (and so we incur an additional $O(\log K)$ in the competitive ratio). The details are given \ifFULL in~\S\ref{sec:nonSimpleDuals}. \fi 
    Blocks which were cheaper than $w(A^t)/(10K)$ do not create any issue because there can be at most $K$ of them, and so their total weight is small in comparison to $w(A^t)$. %
    
    \medskip
     \noindent {\em Short} blocks: Such blocks $B$ straddle two speed classes, say $\ell$ and $\ell+1$, but do not contain too many machines of either class (otherwise they will fall into one of the two categories above). We show \ifFULL in~\S\ref{sec:shortBlocks} \fi that the total weight of such blocks is small compared to $w(A^t)$. The intuitive reason is as follows: for any  two consecutive short blocks $B_1$ and $B_2$, there must be blocks in between them whose span is much longer than $B_2$. Since these blocks freeze before $B_2$, their total weight would be large compared to $w(B_2)$.

In the overall analysis, we charge  short blocks to  simple and long blocks, and  use dual fitting as indicated above to handle simple and long blocks.

\ifFULL
    
\subsection{Defining Dual Variables}
We now give details of the dual variables. 
The $\beta$ dual variables are the easiest to define.  For a time $t$ and  machine $i$ of
speed class $\ell$, define
\[ \beta_{i,t} := \frac{w(A^t)}{ K^2 \cdot \log K \cdot m_\ell} \ .\]
As in~\S\ref{sec:one-star}, we split the $\alpha$ variable across times, but also into 
a ``simple'' part $\alpha_{j,v,t}'$, and a ``non-simple'' part
$\alpha_{j,v,t}''$. The final
$\alpha_{j,v} := \sum_t (\alpha_{j,v,t}' +
\alpha_{j,v,t}'')$. Similarly, we split the $\delta$ variables across
speed classes $\ell$, and into simple/non-simple parts
$\delta'_{j,v, \ell}, \delta''_{j,v,\ell}$, so that 
$\delta_{j,v} := 
\sum_\ell (\delta'_{j,v,\ell} + \delta''_{j,v,\ell}).$ %
Since we are defining two sets of random variables, we will need to show that $\alpha', \delta'$ satisfy~\eqref{eq:dualsum},~\eqref{eq:d2} and~\eqref{eq:dualnew} with  a slack of factor 2 (and similarly for $\alpha'', \delta''$). More formally, we need to check the following conditions:
\begin{align}
    \sum_{v \in T(j)} \delta'_{j,v} & \leq w_j/2 \quad \quad \forall j \label{eq:d2s}  \\
    \label{eq:dualsums}
    \sum_{v \in T(j)} \alpha'_{j,v,t} & \leq w_j/2 \quad \quad \forall j, t \\
     \label{eq:dualnews}
    \alpha'_{j,v,t'} & \leq \frac{\beta_{i,t} \cdot L^{t'}_v}{2s_i}  + \frac{\delta'_{j,v} \cdot L^{t'}_v}{s_i}
    \quad \quad \forall j, v \in T(j), i, t' \leq  t
\end{align}

\subsection{The Simple Dual Variables}
\label{sec:simpledual}
We first define the ``simple'' parts $\delta'$ and $\alpha'$.  We
should think of a job as  being simple with respect to speed
class $\ell$ at time $t$ if its tasks receive speed ``approximately''
$\speed_\ell$. The formal definition is as follows.

\begin{definition} [Simple job] \label{def:simpleJob}
  Call a job $j \in A^t$ to be {\em simple} with respect to speed
  class $\ell$ at time $t$ if the tasks in $T^t(j)$ receive speed\footnote{Recall that by \Cref{cl:rate} all tasks in $T^t(j)$ get the same speed.} in
  the range $[\frac{\gamma \speed_\ell}{64}, 64 \gamma \speed_\ell]$. 
\end{definition}
  
\subsubsection{Definition of Simple Duals} 
  Let
$n^t(j)$ denote the size of $T^t(j)$. For any speed class $\ell$,
let $\last_{\ell,j}$ be the last
time when $j$ is simple with respect to 
$\ell$. %
  Now for each
task $v \in T^{\last_{\ell,j}}(j)$, define
\begin{gather}
  \delta'_{j,v,\ell} := \frac{w_j}{2K \cdot n^{\last_{\ell,j}}(j)} \ .
\end{gather}
For all other tasks $v' \not\in T^{\last_{\ell,j}}(j)$, we set $\delta'_{j,v',\ell}$ values   to 0. 
Define %
$\delta_{j,v}'$ to be
 $\sum_{ \ell} \delta'_{j,v,\ell} $.

Moreover, if $j$ is simple with respect to speed class $\ell$ at time
$t$, then  %
for each task $v \in T^{\last_{\ell,j}}(j)$ define
\begin{gather}
  \alpha'_{j,v,t} := \frac{w_j}{4K \cdot n^{\last_{\ell,j}}(j)}\ .
\end{gather}
Again, the undefined $\alpha'$ duals are set to zero.

\subsubsection{Feasibility of Simple Duals} 
\label{sec:feasibSimpleDuals}

Observe that constraint~(\ref{eq:d2s}) holds because
$ \sum_{ v } \delta'_{j,v}  =  \sum_{v, \ell} \delta'_{j,v,\ell}  \leq w_j/2$.

Feasibility of~\eqref{eq:dualsums} follows similarly.
To show feasibility of constraint~(\ref{eq:dualnews}), we first prove
the following lemma.

\begin{lemma}
  \label{lem:simple}
  For each speed class $\ell$, and time $t$ when $j$ is simple (with respect to any speed class), the following holds for any task $v$ of $j$:
  $$\alpha'_{j,v,t} ~~\leq~~ \frac{1024 \cdot w(A^t) \cdot L_v^t}{K \gamma m_\ell \speed_\ell} + \frac{ \delta'_{j,v} \cdot L^t_v}{\gamma \speed_\ell} \ . $$ 
\end{lemma}

\begin{proof}
  Suppose $j$ is simple with respect to speed class $\ell^\star$ at
  time $t$. Assume that $v \in T^{\last_{\ells,j}}(j)$, otherwise
  $\alpha'_{j,v,t} = 0$ and the lemma trivially holds. 
  Define $\last^\star := \last_{\ells,j}$ to be
  the last time when job $j$ is simple with respect to $\ells$, it follows that
  $\last^\star \geq t$. Let $B$ be the block containing the alive
  tasks of job $j$ at time $\last^\star$.  The set of machines $m(B)$
  corresponding to the block $B$ must contain at least one machine of
  speed class $\ell^\star$ or higher. Indeed, suppose the slowest machine in $m(B)$  belongs to speed class $\ell'  < \ells$. Then every task in $B$ receives rate at least $\gamma \speed_{\ell'} \geq \const \gamma \speed_\ells,$ which violates the condition that $j$ is simple with respect to speed class $\ells$.

 First consider any $\ell < \ell^\star$.  As above, let $B$ be the block containing all alive tasks of $j$ at time $\last^\star$. We apply the second statement of~\Cref{cor:rate1} with $V''=V$ being the set of tasks in $A^{\last^\star}$ which freeze either at the same moment or  before the tasks in $j$ (i.e., those tasks which appear in or before $B$ in the ordering of blocks at time $\last^\star$) and $v$ being any of the tasks in $T^{\last^\star}(j). $ We get 
   $$ \frac{\tw^{\last^\star}(v)}{L^{\last^\star}_v} ~~\leq~~ 
    \frac{\tw^{\last^\star}(V)}{\sum_{v' \in V} L^{\last^\star}_{v'} } ~~\leq~~
  \frac{w(A^{\last^\star})}{\gamma m_\ell \speed_\ell} \ , $$
  where the last inequality follows from the fact that $m(B)$ contains at least one machine of class $\ells$ or higher, which means all the machines of speed class $\ell$ ($<\ell^\star$) are processing tasks from $V$ at time $\last^\star$.

  Since  $\tw^{\last^\star}(v) = \frac{w_j}{n^{\last^\star}(j)}$
  and $n^{t}(j) \geq n^{\last^\star}(j)$, we
  can use the definition of $\alpha'$ to get
  $$ \alpha'_{j,v,t} ~~\leq~~ \frac{w_j}{4K \cdot n^{\last^\star}(j)} ~~\leq~~ \frac{w(A^{\last^\star}) \cdot L^{\last^\star}_v}{4K \cdot \gamma m_\ell 
    \speed_\ell} ~~\leq~~ \frac{1024\, w(A^{t}) \cdot L^{t}_v}{K \gamma m_\ell
    \speed_\ell}\ ,$$ 
    where the last inequality follows from the fact
  that $A^t$ is a superset of $A^{\last^\star},$ and the rate assigned
  to $v$ at time $t$ and $\last^\star$ are within factor
  $64 \times 64$ of each other. 

  Next, suppose $\ell > \ells$. Since $L^t_v \geq \frac{\gamma \speed_\ells}{64} \geq \gamma \speed_\ell,$ 
  and
  $\delta'_{j,v} \geq \alpha'_{j,v,t}$, 
  we get
  $ \alpha'_{j,v,t} \leq \frac{ \delta'_{j,v} L^t_v}{\gamma \speed_\ell} .  $
  
  Finally, suppose $\ell=\ells$. If the block $B$ contains all machines of speed class $\ells$, then the same argument above as for the case when $\ell < \ells$ applies. Otherwise, the block $B$ contains machines of class $\ells$ or smaller. Therefore, the rate assigned to $v$ is at least $\gamma \sigma_\ells = \gamma \sigma_\ell$, and so the argument for the case $\ell > \ells$  above applies. 
\end{proof}

We are now ready to show feasibility of~\eqref{eq:dualnews}. 

\begin{corollary} 
\label{cor:d1}
The solution $(\alpha', \beta, \delta')$ satisfies~\eqref{eq:dualnews}. 
\end{corollary}

 \begin{proof}

Assuming $\gamma \geq 1024 K \log K,$ we see that $\beta_{i,t} \geq \frac{1024 \cdot w(A^t)}{K \gamma m_\ell}$, where $\ell$ denotes the speed class of machine $i$. Since $w(A^t)$, and hence $\beta_{i,t}$, cannot increase as $t$ increases,~\Cref{lem:simple} (applied with $t:=t'$) shows that~\eqref{eq:dualnews} is satisfied.  

\end{proof}

\subsection{The Non-Simple Dual Variables} \label{sec:nonSimpleDuals}

We now define the quantities $\alpha''$ and $\delta''$. 
Before that we need to define ``simple blocks'' and ``long blocks''. For any block $B$, let $s(B)$ denote the total speed of the machines $m(B)$ associated with $B$.

\begin{definition}[Simple block] \label{defn:simple}
We say that block $B$ is {\em simple with respect to speed class $\ell$} at a time $t$ if the {\em average} speed of the tasks in this block, i.e. $\frac{s(B)}{|B|}$,  lies in the range $[ \frac{\gamma \speed_\ell}{2}, 2 \gamma \speed_\ell]$. 
\end{definition}

Note that all the jobs participating in a simple block may not be simple (with respect to the corresponding speed class). Later in \Cref{lem:simpleb}, however, we will show that  a large fraction of such jobs are simple.

\begin{definition}[Long block] \label{defn:long}
We say a block $B$ at time $t$ is {\em long with respect to class $\ell$} if it satisfies the following conditions:
\begin{itemize}
\item It is a non-simple block. 
\item The set $m(B)$  of machines associated with block $B$
contains at least half of the machines of 
class $\ell$, but not all the machines of class $\ell+1$. 
\item $w(B) \geq w(A^t)/(10K)$. 
\end{itemize}
\end{definition}

Observe that for a  long block $B$ as defined above, the corresponding set $m(B)$ can include machines of class less than $\ell$, but will never contain machines of class $\ell+2$ or higher.  

There could also be blocks that are neither simple nor long. We  address them in \S\ref{sec:shortBlocks}.

\subsubsection{Definition of Non-Simple Duals} 
We now define $\delta''_{j,v,\ell}$ values. Consider a task $v$ of a job $j$. Let $t_1, \ldots, t_k$ be the  times $t$ when 
 $v$ belongs to a long block w.r.t. class $\ell$.
Let $B_1, B_2, \ldots, B_k$ be the corresponding long blocks at these times, respectively.  For each speed class $\ell$,  define 
$$ \delta_{j,v,\ell}'' ~:=~ \frac{1}{{96 K   \log K \cdot \tm_\ell}} \cdot \max_{k'=1}^k w(B_{k'}). $$
As before, $\delta''_{j,v} = \sum_\ell \delta''_{j,v, \ell}$ (if any of these quantities is undefined, treat it as 0).

For a job $j$ whose alive tasks belong to a long block $B$ at time $t$, and task $v \in T^t(j)$,  define 
$$\alpha_{j,v,t}'' := \frac{L^t_v \cdot w(B)}{12 K \log K \cdot  s(B)} \ .$$  
Again, if $v$ is not alive at time $t$, $\alpha''_{j,v,t}$ is set to 0. 
 Finally, $\alpha_{j,v}''$ is just the sum of these quantities over all time $t$. 
\subsubsection{Feasibility of Non-Simple Duals}

Now we show that $(\alpha'', \beta, \delta'')$ satisfy~\eqref{eq:dualsums},~\eqref{eq:d2s} and~\eqref{eq:dualnews} (with $\alpha'$ and $\delta'$ replaced by $\alpha''$ and $\delta''$ respectively). 
To prove feasibility of constraint~(\ref{eq:d2s}) in \Cref{lem:rootd}, we need the following claim.

\begin{claim}
\label{cl:long}
Let $B$ be a long block with respect to class $\ell$. Then the total number of tasks in $B$ is at most $\tm_\ell$. Further
the total speed of the machines in $m(B)$ (i.e. $s(B)$),  lies between $\gamma m_\ell \speed_\ell/2$ and $4 \gamma m_\ell \speed_\ell$. 
\end{claim}

\begin{proof}
The average speed of the machines in $m(B)$ satisfies
$$ \frac{s(B)}{|B|} ~~\leq ~~ \frac{\gamma(|B| \speed_{\ell+1} + m_1 \speed_1 + \ldots + m_\ell \speed_\ell)}{|B|} 
~~\leq~~ \frac{\gamma(|B| \speed_{\ell+1} + 2 m_\ell \speed_\ell)}{|B|} \ ,$$
where the inequality follows from the increasing capacity \Cref{assump}. Suppose for contradiction that $|B| \geq \tm_\ell.$ Since $\tm_\ell \speed_{\ell+1} \geq m_\ell \speed_\ell$, the average speed is at most $3 \gamma \speed_{\ell+1}$. Since all machines in $s(B)$ have speed at least $\gamma \speed_{\ell+1}$, the average speed of the machines in $s(B)$ is at least $\gamma \speed_{\ell+1}$. But then  block $B$ is simple with respect to speed class $\ell+1$, which contradicts that $B$ is long with respect to class $\ell$. 

We now prove the second statement. Since $B$ is a long block, $m(B)$ has at least $m_\ell/2$ machines of speed class $\ell$, so $s(B) \geq \gamma m_\ell \speed_\ell/2$. As argued above, $s(B) \leq \gamma (
|B| \speed_{\ell+1} + 2 m_\ell \speed_\ell) \leq \gamma \tm_\ell \speed_{\ell+1} + 2\gamma m_\ell \speed_\ell.$ By definition of $\tm_\ell$ and the increasing capacity \Cref{assump}, 
$\tm_\ell \speed_{\ell + 1} \leq 2m_\ell \speed_\ell.$ This shows that $s(B) \leq 4 \gamma  m_\ell \speed_\ell.$
\end{proof}

Now we can prove  feasibility of constraint~(\ref{eq:d2s}).

\begin{lemma}
\label{lem:rootd}
For any job $j$, we have $\sum_{v \in T(j)} \delta''_{j,v} \leq w_j/2.$
\end{lemma}
\begin{proof}

We begin with a useful claim. 
\begin{claim}
\label{cl:rootd}
Let  $B$ be a block at time $t$ which is long with respect to speed class $\ell$. Suppose $B$ contains all the active tasks of job $j$.   Then 
$ \frac{w(B)}{\tm_\ell} \leq  \frac{6w_j}{n^t(j)} . $ 
\end{claim}

\begin{proof}
By the first statement in \Cref{cor:rate1}, for $V' = B$ and $V$ being the  set of tasks that freeze by the moment $B$ freezes, we get
$\frac{w_j}{n^t(j) \cdot s_{|V|}} \geq \frac{w(B)}{S_{|V|}}.$ Note that  $s_{|V|} \geq \gamma \sigma_{\ell+1}$. 
Since all blocks in $V$ before $B$ contain machines of class $\ell$ or smaller,~\Cref{cl:long} along with the increasing capacity assumption  imply that $S_{|V|} \leq 4 \gamma m_\ell \sigma_\ell + 2 \gamma m_\ell \sigma_\ell = 6 \gamma m_\ell \speed_\ell. $
Thus, $ \frac{w_j}{n^t(j)\cdot \speed_{\ell+1}} \geq \frac{w(B)}{6 m_\ell \cdot \speed_\ell} . $
Therefore, 
$$ \frac{w(B)}{\tm_\ell}   ~~\leq~~  \frac{6 w_j \cdot m_\ell \cdot \speed_\ell}{n^t(j) \cdot \tm_\ell \cdot \speed_{\ell+1}}
~~\leq ~~ \frac{6 w_j}{n^t(j)} \ , $$
where the last inequality follows from the definition of $\tm_\ell$. 
\end{proof}

For the given job $j$,
consider the times $t$ when the tasks in $T^t(j)$ belong to  a long block with respect to class $\ell$ -- let these be $t_1, \ldots, t_k$ (in ascending order) and the corresponding blocks
be $B_{1}, \ldots, B_{k}$. Among these times, starting from $t_1$, we build a subsequence $\tau_1=t_1, \tau_{2}, \ldots, \tau_{u}$ greedily as follows: suppose we have defined $\tau_1, \ldots, \tau_i$, then define $\tau_{i+1}$ to be the smallest index $\tau > \tau_i$ such  that $w(B_\tau) \geq 2 w(B_{\tau_i}). $

Now consider a task $v \in T(j)$. 
Suppose $\delta''_{j,v.\ell}$ is equal to 
$\frac{w(B_{t_r})}{96 K \log K \cdot \tm_\ell}$ where $t_r$ lies between 
$\tau_p$ and $\tau_{p+1}$. Then $w(B_{t_r}) \leq 2 w(B_{\tau_p})$. Further, $v$ must be alive at time $\tau_p$ (since it is alive at time $t_r$) and hence belongs to the block $B_{\tau_p}$ at time $\tau_p$. So, we can upper bound 
$$\delta''_{j,v.\ell} ~~\leq~~ \sum_{s=1}^u \frac{2 w(B_{\tau_s}) \cdot I[v \in A^{\tau_s}]}{96 K  \log K \cdot \tm_\ell}
~~\leq~~ \sum_{s=1}^u \frac{12 w_j \cdot I[v \in A^{\tau_s}]}{96 K \log K \cdot n^{\tau_s}(j)} \ ,$$
where $I[v \in A^t]$ is the indicator variable indicating whether $v$ is unfinished at time $t$, and the last inequality follows from 
Claim~\ref{cl:rootd}. Summing over all tasks $v \in T(j)$, we get 
$$ \sum_{v \in T(j)} \delta_{j,v,\ell}'' ~~\leq~~ \sum_{v \in T(j)} \sum_{s=1}^u \frac{12 w_j \cdot I[v \in A^{\tau_s}]}{96 K  \log K \cdot n^{\tau_s}(j)} ~~=~~ \sum_{s=1}^u 
\sum_{v \in T^{t_s}(j)} \frac{12 w_j }{96 K  \log K \cdot n^{\tau_s}(j)} ~~=~~ 
\frac{ w_j \cdot u}{8 K \log K} \ . $$
Now notice that $u \leq \log (10K)$. The reason is that $B_{\tau_1}$ being a long block implies that $w(B_{\tau_1}) \geq w(A^{\tau_1})/2K$. If $u > \log (10K)$, then $w(B_{\tau_u}) > w(A^{\tau_1}) 
\geq w(A^{\tau_u})$, which is a contradiction. 
Summing the above inequality above all $\ell$ 
proves \Cref{lem:rootd}.
\end{proof}

We now show feasibility of~\eqref{eq:dualsums}. 

\begin{claim}
\label{cl:dualsums}
For any job $j$ and time $t$, 
$$ \alpha''_{j,v,t} \leq w_j/2. $$
\end{claim}
\begin{proof}
 We can assume that $v$ belongs to a long block at time $t$, otherwise $\alpha''_{j,v,t}$ is 0 for all $v$. 
 By definition of rate assignment, for every task $v \in T^t(j)$, $L^t_v =\frac{\tw^t(v) \cdot s(B)}{w(B)}$,  and so, $\sum_{v \in T(j)} L^t_v = \frac{w_j \cdot s(B)}{w(B)}. $ The desired result now follows from the definition of $\alpha''_{j,v,t}.$
\end{proof}

To show feasibility of constraint~(\ref{eq:dualnews}), we first prove
the following lemma.

\begin{lemma}
\label{lem:dang}
Let $v$ be a task of a job $j$. For any time $t$ and machine $i$ belonging to speed class $\ell$,
$$ \alpha''_{j,v,t} ~\leq~  \frac{K \beta_{it} \cdot L^t_v}{6 \gamma \speed_\ell} + \frac{8\delta''_{j,v} \cdot L^t_v}{\gamma \speed_\ell} \ .$$
\end{lemma}
\begin{proof}
We can assume that $v$ is part of a long block $B$ with respect to a speed class $\ells$ at time $t$ (otherwise LHS is 0). 
 First consider the case when $\ell \leq \ells$. Since $s(B) \geq \gamma \cdot m_\ell \speed_\ell/2,$ we see that 
 $$  \alpha''_{j,v,t} ~~ \leq~~ \frac{2L^t_v \cdot w(B)}{12 \gamma K \log K \cdot  m_\ell \speed_\ell} ~~\leq~~ 
 \frac{2L^t_v \cdot w(A^t)}{12 \gamma K \log K \cdot  m_\ell \speed_\ell} ~~\leq~~ 
 \frac{K \cdot L^t_v \cdot  \beta_{it}}{6 \gamma \speed_\ell}\ . $$
 
  So assume $\ell > \ells$. Now 
  $$  \alpha''_{j,v,t}  ~~\leq~~ \frac{2L^t_v \cdot w(B)}{12 \gamma K \log K  \cdot m_\ells \cdot \speed_\ells}
  ~~\leq~~ 
  \frac{2L^t_v \cdot w(B)}{24 \gamma K \log K \cdot \tm_\ells \cdot \speed_{\ells+1}}
  ~~\leq~~ \frac{32 L^t_v \cdot \delta_{j,v}''}{\gamma   \speed_\ell} \,
  $$
  where the second last inequality follows from part~(c) of~\Cref{cl:mlt}, and 
  the last inequality uses $\delta''_{j,v}$ is at least $\delta''_{j,v,\ells}$ and that $\speed_{\ells+1} \geq \speed_\ell$.
\end{proof}

As in the proof of \Cref{cor:d1}, we get the following corollary (assuming $\gamma \geq K)$.
\begin{corollary}
\label{cor:dang}
The solution $(\alpha'', \beta, \delta'')$ satisfies~\eqref{eq:dualnews}. 
\end{corollary}

\Cref{cor:dang} and \Cref{cor:d1} show that $(\alpha:=\alpha'+\alpha'', \beta, \delta:= \delta'+\delta'')$ 
satisfy the dual constraint~(\ref{eq:dualnew}). 

\subsection{Comparing the Objectives}

We start by showing that  a large fraction of jobs participating in a simple block are simple (with respect to the corresponding speed class).

\subsubsection{Handling Simple Blocks}
For a set of tasks $X$, define $w(X)$ as the total weight of the corresponding jobs, i.e., $w(X) =  \sum_{j: T(j) \cap X \neq \emptyset} w_j . $

\begin{lemma}
\label{lem:simpleb}
Let $B$ be a block which is simple with respect to a speed class $\ell$ at time $t$. Let $B'$ be the tasks in $B$ corresponding to jobs which are simple with respect to speed class $\ell$. Then $w(B) \leq 5 w(B').$ 
\end{lemma}

\begin{proof}

 Let $\speed$ denote the average speed assigned to the tasks in $B$. Let $s(B)$ denote the total speed of the machines $m(B)$ associated with $B$, then $\speed = \frac{s(B)}{|B|}. $
 Let $B_1$ be the tasks in $B$ which get speed less than $\speed/32.$ By definition, if $v \in T^t(j) \subseteq B_1$, then 
 $$ L^t_v ~~=~~ \frac{ w_j \cdot s(B)}{n^t(j)\cdot   w(B)} ~~\leq~~ \frac{s(B)}{32|B|} \ .$$
 
 Summing over all the tasks in $B_1$, we see that 
 $w(B_1) \leq w(B)/32$. 

Let $B_2$ be the tasks in $B$ which get speed more than $32 \speed$. By Markov's inequality, $|B_2| \leq |B|/32$. But we really want to upper bound $w(B_2)$. The strategy is as follows: $w(B_2)$ is proportional to the total speed assigned to these tasks, which is at most the total speed of the fastest $|B_2|$ machines in $m(B)$. We now argue that the latter quantity cannot be too large. 

\begin{claim}
The total speed of the fastest $|B|/32$ machines in $m(B)$ is at most $7 s(B)/8.$
\end{claim}
\begin{proof}
 Since $B$ is simple, we know that $\sigma$ lies in the range $[\frac{\gamma \speed_\ell}{2}, 2\gamma \speed_\ell]$. 
The first claim is that at most $2|B|/3$ machines in $m(B)$ can belong to speed class $\ell+1$ or higher. Indeed, otherwise  the average speed of the machines in $m(B)$ would be at most 
$ \frac{\gamma \speed_\ell}{3} + \frac{2 \gamma \speed_{\ell+1}}{3} < \frac{\gamma \speed_\ell}{2}, 
$
because $\sigma_\ell \geq 64 \speed_{\ell+1}.$
 But this is a contradiction.

 Let $\tilde{B}$ be the slowest $31|B|/32$ machines  in $B$. 
 The argument above shows that 
$$s(\tilde{B}) ~~\geq~~ \left(\frac{31}{32} - \frac{2}{3} \right) |B| \gamma \speed_\ell ~~\geq~~ 
|B| \gamma \speed_\ell/4 ~~\geq~~ |B| \speed/8 ~~=~~ s(B)/8 \ ,
$$
which implies the desired claim.
\end{proof}
 Now we bound $w(B_2)$. The above claim implies that 
 $ \frac{7s(B)}{8} \geq \sum_{v \in B_2} L^t_v = \frac{w(B_2)}{w(B)} s(B) .$
 It follows that $w(B_2) \leq \frac{7 w(B)}{8}.$
 Thus, $w(B_1 \cup B_2) \leq 29 w(B)/32.$ The tasks in $B \setminus (B_1 \cup B_2)$ have the property that the jobs corresponding to them are simple  with respect to speed class $\ell$ (\Cref{def:simpleJob}). This implies \Cref{lem:simpleb}. 
\end{proof}

\Cref{lem:simpleb} along with \Cref{cor:d1} imply that the total weight of tasks (or the corresponding jobs) belonging to simple blocks can be accounted for by the dual solution.
We would now like to argue that the total weight of non-simple blocks is upper bounded by that of simple blocks. However this may not be true. {To address this,  we defined long blocks and non-simple duals in  \S\ref{sec:nonSimpleDuals}. We are still not done, since there could be blocks that are neither simple nor long. }

\subsubsection{Handling Cheap Blocks and Short Blocks} \label{sec:shortBlocks}

We  would now like to consider blocks which are neither simple nor long. Notice that in \Cref{defn:long} of a long block, we had added a condition that if $B$ is a long block at time $t$, then $w(B) \geq w(A^t)/(10K)$. It is easy to get rid of this condition. 

\begin{definition}[Cheap block] \label{defn:cheap}
A non-simple block $B$ at a time $t$ is {\em cheap} if  $w(B) < w(A^t)/(10K)$. 
\end{definition}

Note that a cheap block is also  non-long, and it 
 must contain machines of at least two speed classes, else it will be simple. Moreover, 
since there can be at most one block which contains machines of speed class $\ell$ and $\ell+1$ for any index $\ell$, we get the following claim.

\begin{claim}
\label{cl:cheap}
There are at most $K$ cheap blocks. So the total weight of cheap blocks is at most $w(A^t)/10$.
\end{claim}

Finally, we define the notion of a short block. 
\begin{definition}[Short block] \label{defn:short}
A block at time $t$ is {\em short} if it doesn't belong to the category of cheap, long, or simple blocks. 
\end{definition}

We show that the weight of short blocks can be charged  to those of long and simple blocks. 
We  begin by proving some  properties of short blocks.

\begin{claim}
\label{cl:short1}
Any short block contains machines of exactly two consecutive speed classes. 
\end{claim}
\begin{proof}
Let $B$ be a short block, and let $\ell$ be the highest speed class such that $m(B)$ contains a machine of speed class $\ell$. If it contains all the machines of speed class $\ell$, then the average speed
of the machines is at least $\gamma \speed_\ell$ and at most $\frac{\gamma(m_l \speed_\ell+ \ldots + m_1 \speed_1)}{m_\ell} \leq 2 \gamma \speed_\ell$, by the increasing capacity assumption. But then, this is a simple block. 

If it contains a machine of speed class $\ell-2$ as well, then it contains all the machines of speed class $\ell-1$. But then, this a long block with respect to speed class $\ell-1$. 
\end{proof}

We arrange the machines in decreasing order of speed. Each block corresponds to a set of consecutive machines. So we can talk about a left-to-right ordering on the blocks.

\begin{claim}
\label{cl:short2}
Let $B_1$ and $B_2$ be two consecutive short blocks at a particular time $t$. Then the total speed of machines which lie between $m(B_1)$ and $m(B_2)$ is at least 1/8 times that of machines in $m(B_2)$. Further, if $B$ is the left-most short block, then the total speed of machines in $m(B)$ is at most 8 times that of the machines which precede $m(B)$.  
\end{claim}
\begin{proof}
Suppose $B_1$ is to the left of $B_2$. Let $m(B_1)$ contain machines of class $\ellone-1$ and $\ell_1$, and $m(B_2)$ contain machines of class $\ell_2-1$ and $\ell_2$ (Claim~\ref{cl:short1}). Note that $\ell_2 > \ell_1$. 

We first argue that $m(B_1)$ contains at most $m_\ellone/4$ machines of speed class $\ellone$. Suppose not. Then the average speed of the machines in $m(B_1)$ is at most 
$ \frac{\gamma(m_1 \speed_1 + \ldots + m_{\ellone} \speed_\ellone)}{m_{\ellone}/4} 
\leq 8 \gamma \speed_\ellone , $
where the inequality follows from the increasing capacity assumption. 
Since the average speed of these machines is clearly at least $\speed_\ellone$, it follows that $B_1$ is a simple block, which is a contradiction. Since $m(B_2)$ can contain at most $m_{\elltwo-1}/2$ machines of class $\elltwo-1 \geq \ellone$ (otherwise it will be a long block), it follows that there are at least $m_\ellone/4$ machines of speed class $\ellone$ between $m(B_1)$ and $m(B_2)$. The total speed of these machines is $\gamma m_\ellone \speed_\ellone/4$. The increasing capacity assumption implies that the total speed of the machines in $m(B_1)$ is at most $2\gamma  m_\ellone \speed_\ellone.$ This implies the first statement in the claim. The second statement follows in a similar  manner. 
\end{proof}

The above claim implies the following:
\begin{lemma}
\label{lem:counting}
Let $B_1$ and $B_2$ be two consecutive short  blocks at time $t$. Let $B$ be the set of tasks belonging to the blocks lying between these two blocks. Then 
$w(B_2) \leq 8 \cdot w(B))$. Similarly, if  $B_1$ is the left-most short expensive block and $B$ is the set of tasks belonging to the blocks lying to the left of $B_1$, then $w(B_2) \leq 8 \cdot w(B))$.
\end{lemma}

\begin{proof}
We apply the second statement in Corollary~\ref{cor:rate1}, with $V''$ being the set of tasks in $B$, and $v$ is a task in $B_2$. We get 
$ \frac{\tw^t(v)}{L^t_v} ~\leq~ \frac{w(B_2)}{s(B_2)} . $
This implies that 
$$ \frac{w(B_2)}{s(B_2)} ~~\leq~~ \frac{\sum_{v \in B_2} \tw^t(v)}{\sum_{v \in B_2} L^t_v} ~~\leq~~ \frac{w(B_2)}{s(B_2)} \ . $$
The first statement in the lemma now follows from Claim~\ref{cl:short2}. The second statement follows similarly. 
\end{proof}

We are now ready to show that the total weight of jobs in $A^t$ is dominated by those belonging to simple or long blocks.
 Recall, we have defined four kinds of blocks -- let $\cB^t_{\textsf{simple}}$, $\cB^t_{\textsf{long}}$, $\cB^t_{\textsf{short}}$, and $\cB^t_{\textsf{cheap}}$ denote the set of simple, long, short, and cheap blocks at time $t$, respectively. 

\begin{lemma}
\label{lem:final}
We have $w(A^t) \leq 90 \cdot   \big(w(\cB^t_{\textsf{simple}}) + w(\cB^t_{\textsf{long}})\big)$.  
\end{lemma}

\begin{proof}
 Consider a time $t$. 
 For a set $\cB$ of blocks, define $w(\cB)$ as the total weight of the blocks in it. 
 We know from \Cref{cl:cheap} that $w(\cB^t_{\textsf{cheap}}) \leq w(A^t)/10$. Further, Lemma~\ref{lem:counting} implies that $w(\cB^t_{\textsf{short}}) \leq 8\big(w(\cB^t_{\textsf{simple}}) + w(\cB^t_{\textsf{long}}) + w(\cB^t_{\textsf{cheap}}) \big).$ Combining, we get that $ w(A^t) =  w(\cB^t_{\textsf{simple}}) + w(\cB^t_{\textsf{long}}) + w(\cB^t_{\textsf{cheap}}) + w(\cB^t_{\textsf{short}})$ satisfies
\[
 w(A^t)  ~~\leq ~~  9 \big(w(\cB^t_{\textsf{simple}}) + w(\cB^t_{\textsf{long}})\big) + 9 w(\cB^t_{\textsf{cheap}}) ~~
 \leq~~  9 \big(w(\cB^t_{\textsf{simple}}) + w(\cB^t_{\textsf{long}})\big) + \frac{9 w(A^t)}{10}  \  . 
\]
\end{proof}

\subsubsection{Putting it Together}
Let $\cA$ denote our algorithm and $C^\cA$ denote the sum of weighted completion time of the jobs. We first relate $\beta_{i,t}$ variables with $C^\cA$.
\begin{claim} \label{cl:beta} The duals $\beta_{i,t}$  satisfy
$$ \sum_{i,t} \beta_{i,t} ~=~ \frac{C^\cA}{cK \log K}\ . $$
\end{claim}
\begin{proof}
 This follows directly from the definition of $\beta_{i,t}.$ Let $M_\ell$ denote the machines of class $\ell$. Then 
 $$ \sum_{i,t} \beta_{i,t} ~~=~~ \sum_t \sum_\ell \sum_{i \in M_\ell} \frac{w(A^t)}{cK^2 \log K \cdot m_\ell} ~~=~~ \sum_t \frac{w(A^t)}{cK \log K}\ .$$
 
 Since $C^\cA = \sum_t w(A^t)$, the desired result follows. 
\end{proof}

We now relate the $\alpha_{j,v}$ variables to the objective function. 
\begin{claim}
\label{cl:alpha}
We have
$$ \sum_{j,v} \alpha_{j,v} ~\geq~ \frac{C^\cA}{1800 \cdot K \log K}\ . $$
\end{claim}

\begin{proof}
 Let $\cB_{\textsf{simple}}^t$ and $\cB_{\textsf{long}}^t$ denote the set of simple and long blocks respectively at time $t$. For a set $\cB$ of blocks at a particular time $t$ and job $j$, let $I[j \in \cB]$ be the indicator variable denoting whether the alive tasks of $j$ belong to a block in $\cB$. 
 
 We first consider the $\alpha'_{j,v}$ variables. Consider a time $t$ such that $j$ is simple with respect to a speed class $\ell$ at this time. Then, by definition, 
 $ \sum_{v \in T^t(j)} \alpha'_{j,v,t} = \frac{w_j}{4K} . $
 Therefore,
 $$ \sum_{j,v} \alpha'_{j,v}~~ =~~ \sum_t \sum_{j,v} \alpha'_{j,v,t}~~ =~~ 
 \sum_t \sum_j \frac{w_j I[j \in \cB^t_{\textsf{simple}}]}{4K}~~ \geq~~ \sum_t \frac{w(\cB^t_{\textsf{simple}})}{20K}\ , $$
 where the last statement follows from Lemma~\ref{lem:simpleb}.

 We now consider $\alpha_{j,v}''$ variables. Recall that if the tasks $T^t(j)$ of a job $j$ belong to a long block $B$ at a time $t$, then $\alpha''_{j,v,t}$ is defined as $\frac{L^t_v \cdot w(B)}{12 K \log K s(B)}$ for every task $v \in T^t(j)$. 
 Since $L^t_v = \frac{\tw^t(v) \cdot s(B)}{w(B)}$, we get $\sum_{v \in T^t(j)} \alpha''_{j,v,t} = \frac{w_j}{12 K \log K} . $
 Thus, 
 $$ \sum_{j,v} \alpha''_{j,v} ~~=~~ \sum_t \sum_{j,v} \alpha''_{j,v,t} ~~=~~ 
 \sum_t \frac{w_j \cdot I[j \in \cB^t_{\textsf{long}}]}{12 K \log K} ~~=~~ \sum_t \frac{w(\cB^t_{\textsf{long}})}{12 K \log K}\ . $$
 
 Combining the above two observations with Lemma~\ref{lem:final}, we get 
$ \sum_{j,v} \alpha_{j,v} \geq \frac{w(A^t)}{1800 \cdot K \log K} .  	$
\end{proof}

We are now ready to bound the competitive ratio of Algorithm $\cA$. 
\begin{theorem}
 \label{thm:final}
 The algorithm $\cA$ is $O(K^3 \log^2 K)$-competitive. 
\end{theorem}

\begin{proof}
 We know that $(\alpha:=\alpha'+\alpha''~,~ \beta ~,~ \delta:= \delta'+\delta'')$  is a feasible dual LP solution:
 \Cref{cor:dang} along with \Cref{cor:d1} show 
feasibility of dual constraint~(\ref{eq:dualnew}), and  \Cref{lem:rootd} along with \S\ref{sec:feasibSimpleDuals} prove feasibility of  
dual constraint~(\ref{eq:d2}). Similarly,~\ref{cl:dualsums} and the definition of $\alpha'$ shows the feasibility of~\Cref{eq:dualsum}. 
 So the objective function value is a valid lower bound on the objective value of the optimal solution. By Claim~\ref{cl:alpha} and Claim~\ref{cl:beta}, we see that 
 $$ \sum_{j,v} \alpha_{j,v} - \sum_{i,t} \beta_{i,t} ~=~ \Omega
 \left( \frac{C^\cA}{K \log K} \right) \  . $$
 Therefore, $C^\cA$ is at most $O(K \log K)$ times the optimal cost. However, we had also assumed a speed-up of $\gamma = O(K \log K)$, and another speed-up of $O(K)$ for the increasing capacity \Cref{assump}. Combining everything, we see that the algorithm is $O(K^3 \log^2 K)$-competitive. 
\end{proof}

\subsection{Extension to Arbitrary Release Dates}
\label{sec:genrj}
We now highlight the main steps needed to extend the above analysis to the more general case when jobs can have arbitrary release dates. Recall that for a time $t$, $A^t$ denotes the set of jobs which have not finished processing at time $t$. We extend this definition to the general case by defining $A^t$ to be the set of jobs which have been released by time $t$, but haven't finished processing yet. A related quantity, $U^t$, is the set of all jobs which have not finished processing till time $t$ -- this includes $A^t$ and the jobs which have not been released till time $t$. Note that the total weighted completion time of a schedule is $\sum_t w(U^t).$ 

The definition of $\beta_{i,t}$ remains unchanged except that we replace $w(A^t)$ by $w(U^t)$. The definitions of $\delta'$ and $\delta''$ do not change. For a job $j$, a task $v \in T(j)$ and time $t \geq r_j$, $\alpha_{j,t}$ is 

\else 

\fi

\section{Discussion}
Several interesting problems remain open. (i) Can we close the gap
between lower bound of $\Omega(K)$ and upper bound of $O(K^3 \log^2 K)$?
(ii) Can we prove an analogous result for weighted \emph{flow-time}
(with speed augmentation)? (iii) Can we generalize this result to the
unrelated machines setting? (iv) Our lower bound of
$\Omega(K)$-competitive ratio relies on non-clairvoyance; can we prove
a better bound if the processing times of tasks are known at their
arrival times? 

{\small \bibliography{bib}}

\newcommand{\etalchar}[1]{$^{#1}$}
\begin{thebibliography}{BMP{\etalchar{+}}10}

\bibitem[AC08]{4536445}
C.~{Anglano} and M.~{Canonico}.
\newblock Scheduling algorithms for multiple bag-of-task applications on
  desktop grids: A knowledge-free approach.
\newblock In {\em 2008 IEEE International Symposium on Parallel and Distributed
  Processing}, pages 1--8, 2008.

\bibitem[ALLM16]{ALLM}
Kunal Agrawal, Jing Li, Kefu Lu, and Benjamin Moseley.
\newblock Scheduling parallel {DAG} jobs online to minimize average flow time.
\newblock In {\em Proceedings of SODA}, pages 176--189, 2016.

\bibitem[BMP{\etalchar{+}}10]{BenoitMPRV10}
Anne Benoit, Loris Marchal, Jean{-}Francois Pineau, Yves Robert, and
  Fr{\'{e}}d{\'{e}}ric Vivien.
\newblock Scheduling concurrent bag-of-tasks applications on heterogeneous
  platforms.
\newblock {\em {IEEE} Trans. Computers}, 59(2):202--217, 2010.

\bibitem[BN15]{BNF15}
Abbas Bazzi and Ashkan Norouzi{-}Fard.
\newblock Towards tight lower bounds for scheduling problems.
\newblock In {\em Proceedings of ESA}, pages 118--129, 2015.

\bibitem[CS99]{ChudakS99}
Fabi{\'{a}}n~A. Chudak and David~B. Shmoys.
\newblock Approximation algorithms for precedence-constrained scheduling
  problems on parallel machines that run at different speeds.
\newblock {\em J. Algorithms}, 30(2):323--343, 1999.

\bibitem[CSV09]{CorreaSV09}
Jos{\'{e}}~R. Correa, Martin Skutella, and Jos{\'{e}} Verschae.
\newblock The power of preemption on unrelated machines and applications to
  scheduling orders.
\newblock In {\em Proceedings of {APPROX}/{RANDOM}}, pages 84--97, 2009.

\bibitem[GGKS19]{GGKS-ICALP19}
Naveen Garg, Anupam Gupta, Amit Kumar, and Sahil Singla.
\newblock Non-clairvoyant precedence constrained scheduling.
\newblock In {\em Proceedings of ICALP}, pages 63:1--63:14, 2019.

\bibitem[HSSW97]{Hall}
Leslie~A. Hall, Andreas~S. Schulz, David~B. Shmoys, and Joel Wein.
\newblock Scheduling to minimize average completion time: off-line and on-line
  approximation algorithms.
\newblock {\em Math. Oper. Res.}, 22(3):513--544, 1997.

\bibitem[Li17]{Li17}
Shi Li.
\newblock Scheduling to minimize total weighted completion time via
  time-indexed linear programming relaxations.
\newblock In {\em Proceedings of FOCS}, pages 283--294. 2017.

\bibitem[MK15]{MoschakisK15}
Ioannis~A. Moschakis and Helen~D. Karatza.
\newblock Multi-criteria scheduling of bag-of-tasks applications on
  heterogeneous interlinked clouds with simulated annealing.
\newblock {\em J. Syst. Softw.}, 101:1--14, 2015.

\bibitem[MQS98]{MQS}
Alix Munier, Maurice Queyranne, and Andreas~S. Schulz.
\newblock Approximation bounds for a general class of precedence constrained
  parallel machine scheduling problems.
\newblock In {\em Proceedings of IPCO}, volume 1412, pages 367--382. 1998.

\bibitem[QS01]{QueyranneS01}
Maurice Queyranne and Maxim Sviridenko.
\newblock A (2+epsilon)-approximation algorithm for generalized preemptive open
  shop problem with minsum objective.
\newblock In {\em Proceedings of {IPCO}}, volume 2081, pages 361--369, 2001.

\bibitem[RS08]{RS}
Julien Robert and Nicolas Schabanel.
\newblock Non-clairvoyant scheduling with precedence constraints.
\newblock In {\em Proceedings of SODA}, pages 491--500, 2008.

\end{thebibliography}

\appendix
\ifFULL
\section{Missing Proofs of~\Cref{sec:intro}}
 \label{sec:AppendixOne}
 
 Here we sketch why for general DAGs on related machines, every non-clairvoyant algorithm has a large competitive ratio.
 
 \begin{theorem} 
 Any non-clairvoyant algorithm for related machines DAG scheduling must have $\Omega \big( \frac{\log m}{\log \log m} \big)$-competitive ratio. 
 \end{theorem}
 \begin{proof}{\em(Sketch)} 
 Consider a single DAG with $m$ nodes, which is formed by a $1/\epsilon$-ary tree where every edge is directed away from the root (so, for instance, every job depends on the root-job) for $\epsilon= \Theta(\frac{\log\log m}{\log m})$. The objective is to minimize the makespan. 
 We are given two kinds of machines: there is one fast machine of speed $s_1=1$ and there are $m$ slow machines of speed $s_2 = \epsilon$. We will set the job lengths such that the offline optimum makespan is $1/\epsilon$ but any non-clairvoyant algorithm has makespan  $\Omega(1/\epsilon^2)$.

 To set the job lengths, consider a random root-leaf path $P$ on the DAG. Now every node/job on $P$ or incident onto $P$ (i.e., those that share an edge with a node on $P$) has length $1$, and every other node/job in the DAG has length $0$. The offline optimum is $1/\epsilon$ because the fast machine can work on the unit-sized jobs on $P$ and  the slow machines can work on the unit-sized jobs incident onto $P$. However, any non-clairvoyant algorithm (which does not know $P$) will spend  $\Omega(1/\epsilon)$ time at each node of $P$ to identify the next node of $P$, and hence has makespan  $\Omega(1/\epsilon^2)$.
 \end{proof}
 
\else
\fi

 \section{Missing Proofs of~\Cref{sec:stmt-lower-bound}}
 \label{sec:missing2}
 \incap*
\begin{proof}%
We show how to transform the instance so that it satisfies the increasing 
capacity assumption, while losing only $O(K)$-factor in the competitive ratio. For sake of brevity, let $\kappa$ denote the constant 64. 

 For a speed class $\ell$, let $C_\ell$ denote $m_\ell \speed_\ell$, i.e., the total processing capacity of the machines in this speed class.  
  Starting from speed class $1$, we construct a subset $X$ of speed classes as follows: if $\ell$ denotes the last speed class added to $X$, then
  let $\ell' > \ell$ be the smallest class such that $C_{\ell'} \geq 2 \kappa C_{\ell}.$ We add $\ell'$ to $X$ and continue this process till we have exhausted all the speed classes. 
  
  Consider the instance $\I'$ in which the set of jobs is the same as those in $\I$, but there are $K m_\ell$ machines of speed class $\ell$ for each $\ell \in X$. For a speed class $\ell \in X$, let $C_\ell'$ denote 
  $ 2 \kappa K m_\ell \speed_\ell,$ which is at most the total capacity of the speed class $\ell$ machines in $\I'$. 
  Let us now consider the optimal solutions of the two instances. 
 We first observe that $\opt(\I') \leq \opt(\I)$. Consider two consecutive speed classes $\ell_1 < \ell_2$ in $X$. From the definition of $X$, we see that $C_{\ell_1}' \geq \sum_{l=\ell_1}^{\ell_2-1} C_l.$ Therefore all the processing done by a solution to $\I$ on machines of speed class $[\ell_1, \ell_2)$ during a timeslot $[t,t+1]$ can be performed on machines of speed class $\ell_1$ in $\I'$ during the same timeslot. Therefore, $\opt(\I') \leq \opt(\I)$. 
 
 For the converse statement, it is easy to see that if we give $2 \kappa K$ speedup to each machine in $\I$, then the processing capacity of each speed class in $\I$ is at least that in $\I'$. Therefore, $\opt(\I) \leq 2 \kappa K \opt(\I')$. Therefore, replacing $\I$ by $\I'$ will result in $O( \kappa K)$ loss in competitive ratio. It is also easy to check that $\I'$ satisfies increasing capacity assumption. 
 
 Observe that the conversion from $\I$ to $\I'$ can be easily done at the beginning -- we just need to identify the index set $X$, and use only these for processing. The factor $K$ loss in competitive ratio is also tight for the instance $\I$ where all speed classes have the same capacity. 
\end{proof}

\section{Missing proofs of~\Cref{sec:scheduling-algorithm}}
\label{sec:missing3}

\feas*
\begin{proof}%
  The rates assigned to  tasks change only when one of these events happen: (i) a new job $j$ arrives, (ii) an existing task finishes. 
  Assuming that the job sizes, release dates are integers, we can find a suitable $\delta > 0$ (which will also depend on the speeds of the machines) such that all the above events happen at integral multiples of $\delta$. 
  
  Consider an interval $[t, t+\delta),$ where $t$ is an integral multiple of $\delta.$ We need to show that if $L^t_v$'s satisfy the condition~($\star$), then we can build a feasible schedule during $[t, t + \delta)$. By feasibility, we mean that each task $v$ can be processed to an extent of $\barp_v := L^t_v \cdot \delta$ extent and at any point of time, it gets processed on at most one machine. 
  
  We follow a greedy strategy to build the schedule. Suppose we have built the schedule till time $t' \in[t, t+\delta)$. At time $t'$, we order the tasks in descending order of the remaining processing requirement for this slot (at time $t$, each task $v$ has processing requirement of $\barp_v$). Let the ordered  tasks at time $t'$ be $v_1, \ldots, v_n.$ We schedule $v_i$ on machine $i$. 
  
  Suppose for the sake of contradiction, a task $\vst$ is not able to complete $\barp_{\vst}$ amount of processing. We first make the following observation: 
  \begin{claim}
    \label{cl:speed1}
    Let $v$ and $v'$ be two tasks such that at some time $t' \in [t, t+\delta)$, we prefer $v$ to $v'$ in the ordering at time $t'$. Then if $v'$ does not complete $\barp_{v'}$ amount of processing during $[t, t+\delta)$, then neither does $v$. 
  \end{claim}
  \begin{proof}
    Since we prefer $v$ at time $t'$, $v$ has more remaining processing time. If we never prefer $v'$ to $v$ after time $t'$, then $v$ always has more remaining processing requirement than $v'$ during this interval. If we prefer $v'$ to $v$ at some point of time during $(t',t+\delta)$, then it is easy to check that the remaining processing requirements for both $v$ and $v'$ will remain the same. The result follows easily from this observation. 
  \end{proof}
  
  Starting from $\{\vst\}$, we build a set $S$ of tasks which has the following property: if $v \in S$, then we add to $S$ all the tasks $v'$ such that $v'$ was preferred over $v$ at some point of time during $[t, t+\delta)$. Repeating application of Claim~\ref{cl:speed1} shows that none of these tasks $v$ complete $\barp_v$ amount of processing during $[t, t+\delta)$. Let ${\bar m}$ denote $|S|$. We note that only tasks in $S$ would have been processed on the first ${\bar m}$ machines during $[t, t+\delta)$ -- otherwise, we can add more tasks to $S$. Since none of these tasks finish their desired amount of processing during this interval, it follows that 
  $$ \sum_{v \in S} \barp_v \geq \gamma \delta \cdot S_{{\bar m}} \ . $$
  Since $\barp_v = \delta L^t_v $, we see that the set of tasks in $S$ violates~($\star$). This is a contradiction, and so such a task $\vst$ cannot exist. 
\end{proof}

\rate*
\begin{proof}
  For any task $v'$, let $\tau_{v'}$ be the value of $\tau$ at which
  $v'$ freezes. We know that
  \begin{gather}
    \sum_{v' \in V} \tw^t(v') \cdot \tau_{v'} = \gamma S_{|V|}\ . \label{eq:1}
  \end{gather}
  Since
  $ \sum_{v' \in V\setminus \{v\}} \tw^t(v') \cdot\tau_{v'} \leq \gamma S_{|V|-1}$
  by feasibility, it follows that
  \begin{gather}
    \tw^t(v) \cdot \tau_v \geq \gamma s_{|V|}\ . \label{eq:2}
  \end{gather}
  Now for all $v' \in V'$, we have $\tau_{v'} \geq \tau_v$,
  so
  \[ \tw^t(V') \cdot \tau_v ~~=~~ \sum_{v' \in V'} \tw^t(v') \cdot \tau_v ~~\leq~~
  \sum_{v' \in V'} \tw^t(v') \cdot\tau_{v'} ~~\leq~~
  \sum_{v' \in V} \tw^t(v') \cdot \tau_{v'} ~~\stackrel{(\ref{eq:1})}{=}~~ \gamma S_{|V|}\ . \] 
  Hence,  the first claim follows: 
    \[ \frac{\tw^t(v)}{s_{|V|}} ~~\stackrel{(\ref{eq:2})}{\geq}~~ \frac{\gamma}{\tau_{v}} ~~\geq~~ \frac{\tw^t(V')}{S_{|V|}}\ . \] 
For the second claim,
  $$ \sum_{v' \in V''} L_{v'}^t ~~=~~ \sum_{v' \in V''} \tw^t_{v'} \cdot
  \tau_{v'} ~~\leq~~  \tw^t(V'') \cdot \tau_v \ .$$
  The claim now follows by the definition $L^t_v = \tw^t(v) \cdot \tau_v$.
\end{proof}

\section{Missing Proofs of~\Cref{sec:lp}}
\label{sec:appendixlp}
\lp*
\begin{proof} %
Consider a schedule $\calS$, and let $x_{ivt}$ be the extent of processing done on a task $v$ (belonging to job $j$) during $[t,t+1]$ on machine $i$. More formally, if the task is processed for $\varepsilon$ units of time on machine $i$ during this time slot, then we set $x_{ivt}$ to $\varepsilon \cdot s_i.$ Constraint~\eqref{eq:alphaDual} states that every task $v$ needs to be processed to an extent of $p_v$, whereas~\eqref{eq:betaDual} requires that we cannot do more than $s_i$ unit of processing in a unit time slot on machine $i$. Now we verify verify~\eqref{eq:deltaDual}. Consider a task job $j$ and a task $v$ belonging to it. The total processing time of $v$ is 
\begin{align}
    \label{eq:flow1}
    \sum_{i,t} \frac{x_{ivt}}{s_i}. 
\end{align}
The completion time $F_j$ of $j$ is at least the processing time of each of the tasks in it. Finally, we check~\eqref{eq:deltatdual}. Define $F_{j,t}$ to be 1 if $j$ is alive at time $t$. The RHS of this constraint is the fraction of $v$ which is done after time $t$; and so if this is non-zero, then $F_{j,t}$ is 1. This shows the validity of this constraint.   

In the objective function, the first term is the total weighted completion time of all the jobs. The second term is also the same quantity, because $F_j$ is equal to $\sum_{t \geq r_j} F_{j,t}.$
\end{proof}

\end{document}